\documentclass[format=acmsmall]{acmart}
\usepackage{acm-arxiv}
\usepackage{booktabs} 

\usepackage{graphicx}
\usepackage{xcolor}
\usepackage{enumerate}
\usepackage{color}
\definecolor{darkblue}{rgb}{0.0,0.0,0.2}
\definecolor{gray}{gray}{0.5}
\definecolor{lightred}{rgb}{1,0.6,0.6}
\definecolor{darkgreen}{rgb}{0,0.5,0}

\usepackage{tikz}  
\usepackage{tikzsymbols}
\usetikzlibrary{shapes.callouts} 
\usetikzlibrary{math,calc}
\usetikzlibrary{decorations.pathmorphing,patterns,snakes}

\newcommand{\Comments}{0}
\usepackage[colorinlistoftodos,textsize=tiny]{todonotes}
\newcommand{\mynote}[2]{\ifnum\Comments=1\textcolor{#1}{#2}\fi}
\newcommand{\mytodo}[2]{\ifnum\Comments=1
  \todo[linecolor=#1!80!black,backgroundcolor=#1,bordercolor=#1!80!black]{#2}\fi}

\newcommand{\btw}[1]{}

\ifnum\Comments=1               
\else
\fi

\usepackage{dsfont}
\usepackage[font={small}]{caption}
\usepackage{booktabs}
\usepackage[shortlabels]{enumitem}

\usepackage{etoolbox,xparse,xspace,etex}
\usepackage{array,multirow,booktabs} 
\usepackage{mathtools}
\newcommand\numberthis{\addtocounter{equation}{1}\tag{\theequation}}

\usepackage{thm-restate}
\newcommand{\restatehack}[1]{}

\newcommand{\Tr}{\mathrm{Tr}}
\newcommand{\Diag}{\mathrm{Diag}}

\renewcommand{\C}{\mathbb{C}}

\newcommand{\E}{\mathop{\mathbb{E}}}

\newcommand{\N}{\mathbb{N}}
\renewcommand{\P}{\mathcal{P}}

\newcommand{\R}{\mathcal{R}}
\newcommand{\V}{\mathcal{V}}
\newcommand{\X}{\mathcal{X}}
\newcommand{\Y}{\mathcal{Y}}

\newcommand{\simplex}{\Delta_\Y}

\newcommand{\convhull}{\mathrm{conv}}

\newcommand{\dom}{\mathrm{dom}}
\newcommand{\ones}{\mathds{1}}

\def\reals{\mathbb{R}}

\def\extreals{\mathbb{\overline{R}}}

\newcommand{\argmin}{\mathop{\mathrm{argmin}}}
\newcommand{\argmax}{\mathop{\mathrm{argmax}}}

\newcommand{\inprod}[2]{\left\langle #1, #2 \right\rangle}

\newcommand{\toto}{\rightrightarrows}
\newcommand{\elic}{\mathrm{elic}}

\newcommand{\lin}{\mathrm{Lin}}
\newcommand{\extlin}[1]{\mathrm{Lin}(#1\to\extreals)}
\newcommand{\Pos}{\mathrm{Pos}}
\newcommand{\Herm}{\mathrm{Herm}}
\newcommand{\Dens}{\mathrm{Dens}}
\newcommand{\Meas}{\mathrm{Meas}}
\newcommand{\extHerm}{\overline{\mathrm{Herm}}}
\newcommand{\mufix}{{\mu^{\diamond}}}
\newcommand{\mufixy}{{\mu_y^{\diamond}}}
\newcommand{\Sfix}{{S^{\diamond}}}
\newcommand{\Gammafix}{{\Gamma^{\diamond}}}
\newcommand{\Pfix}{{\P^{\diamond}}}
\newcommand{\muspec}{\mu^{\mathrm{spec}}}
\newcommand{\extsub}{\overline{\partial}}
\newcommand{\id}{{\Gamma_{\mathrm{id}}}}
\newcommand{\idfix}{{\Gamma_{\mathrm{id}}^{\diamond}}}
\newcommand{\brier}{s_{\textrm{Brier}}}

\newcommand{\elicq}{{\elic_Q}}
\newcommand{\elicfix}{{\elic_{\Pfix}}}

\newcommand{\nsphere}{{\mathcal{S}^{n-1}}}
\newcommand{\subselect}[3]{\ensuremath{\{#1_{#2}\}\in\extsub #3}}
\newcommand{\vstar}[2]{{#1^*}_{\hspace*{-.45em}#2}}
\renewcommand{\vvstar}[2]{#1_{#2}\vstar{#1}{#2}}

\usepackage{fancy-matrix-thingy}

\title{Quantum Information Elicitation}
\author{Rafael Frongillo}
\begin{abstract}
  In the classic scoring rule setting, a principal incentivizes an agent to truthfully report their probabilistic belief about some future outcome.
  This paper addresses the situation when this private belief, rather than a classical probability distribution, is instead a quantum mixed state.
  In the resulting quantum scoring rule setting, the principal chooses both a scoring function and a measurement function, and the agent responds with their reported density matrix.
  Several characterizations of quantum scoring rules are presented, which reveal a familiar structure based on convex analysis.
  Spectral scores, where the measurement function is given by the spectral decomposition of the reported density matrix, have particularly elegant structure and connect to quantum information theory.
  Turning to property elicitation, eigenvectors of the belief are elicitable, whereas eigenvalues and entropy have maximal elicitation complexity.
  The paper concludes with a discussion of other quantum information elicitation settings and connections to the literature.
\end{abstract}

\begin{document}

\begin{titlepage}
\maketitle
\end{titlepage}

\section{Introduction}

In the field of information elicitation, one is interested in scoring mechanisms which incentivize self-minded agents to reveal their private information.
Perhaps the most fundamental scenario is the scoring rule setting, where a principal wishes to incentivize a single agent to reveal their private belief about the probability of a future outcome.
Specifically, given a reported probability distribution $p\in\simplex$, and the realized outcome $y\in\Y$, the principal wishes to design a \emph{proper} scoring rule $S(p,y)$, such that the agent will maximize their expected score by reporting $p$ to be their true belief.
This basic setting forms the foundation of more complex settings, such as property elicitation, where a summary statistic of the belief is sought, and multi-agent settings like prediction markets, wagering mechanisms, and peer prediction.

In many ways, the fields of information elicitation and information theory grew up together.
Just two years after Shannon's groundbreaking 
work, Brier introduced the first proper scoring rule in 1950~\cite{brier1950verification}.
In another two years, I.J.\ Good introduced the log score~\cite{good1952rational}, apparently unaware that its expected score function is none other than (negative) Shannon entropy.
In 1956, in a paper communicated by Shannon, McCarthy made this connection between information theory and information elicitation explicit, already observing the prominent role that convexity plays~\cite{mccarthy1956measures}.

Since the turn of the century, however, another type of information is becoming increasingly relevant in theoretical computer science~\cite{lo1998introduction,nielsen2002quantum,aaronson2013quantum}, machine learning~\cite{aaronson2007learnability,koolen2011learning,qi2013quantum,gao2018experimental}, and game theory~\cite{lassig2002quantum,zhang2012quantum,bostanci2021quantum}: quantum information.
As quantum information theory is a thriving field,
it is natural to ask, what of quantum information elicitation?
That is, what if the private infomation to be elicited is not a classically random distribution, but instead a quantum mixed state?
This situation arises, for example, when applying learning algorithms to data generated from quantum states~\cite{qi2013quantum,gao2018experimental}.
One interesting complication is that, from a physical standpoint, any sample ``drawn'' from a quantum mixed state must be mediated by a choice of measurement.
Roughly speaking, the measurement procedure takes a particular projection of the density matrix onto a classical distribution over some set of outcomes, from which the actual outcome of the measurement is drawn (Fig.~\ref{fig:overview}).
How then can the principal incentivize an agent to reveal their entire mixed state, when given only measurement access?

\begin{figure}[b]
  \definecolor{belief}{rgb}{0,0,.5}
  \definecolor{report}{rgb}{.3,.5,0}
  \definecolor{mixed}{rgb}{.3,.5,.5}
  
  \begin{tikzpicture}
    \node (dude) at (0,0) {\Strichmaxerl[4]};
    \node (score) at (3,0) {$s(\,\textcolor{report}{p'},\textcolor{belief}{y}\,)$};
    \node[draw,
    align=center,
    cloud callout,
    cloud puffs = 10,
    aspect=2,
    callout absolute pointer={(dude.east)+(0,0.3)}] (p) at (1.5,0.8) {$\textcolor{belief}{p}$};
    \node[draw,
    align=center,
    ellipse callout,
    callout absolute pointer={(dude.east)}] (pp) at (1.5,-0.8) {$~~~~~\textcolor{report}{p'}~~~~~$};
    \draw [thick,->,color=belief,
    decorate, decoration={
      snake,
      segment length=4,
      amplitude=0.8,
      post=lineto,
      post length=2pt
    }]  (p.center)+(5pt,-1pt) to [out=0,in=90] ($(score.center)+(0.3,0.2)$);
    \draw [thick,->,color=report] (pp.center)+(5pt,0pt) to [out=30,in=-90] ($(score.center)+(-0.1,-0.2)$);
  \end{tikzpicture}
  \hspace{0.8cm}
  \begin{tikzpicture}
    \node (dude) at (0,0) {\Strichmaxerl[4]};
    \node (inprod) at (3.5,0.4) {$\!\!\inprod{\mu(\textcolor{report}{\rho'})}{\textcolor{belief}{\rho}}\!\!$};
    \node (score) at (3.2,-0.5) {$s(\textcolor{report}{\rho'},\textcolor{mixed}{y})$};

    \node[draw,
    align=center,
    cloud callout,
    cloud puffs = 10,
    aspect=2,
    callout absolute pointer={(dude.east)+(0,0.3)}] (p) at (1.5,0.8) {$\textcolor{belief}{\rho}$};
    \node[draw,
    align=center,
    ellipse callout,
    callout absolute pointer={(dude.east)}] (pp) at (1.5,-0.8) {$~~~~~\textcolor{report}{\rho'}~~~~~$};

    \draw [thick,->,color=belief] (p.center)+(5pt,0pt) to [out=0,in=90] ($(inprod.center)+(0.5,0.2)$);
    \draw [thick,->,color=report] (pp.center)+(5pt,2pt) to [out=40,in=100] ($(inprod.center)+(-0.1,0.2)$);

    \draw[thick,color=mixed,
    decoration={
      brace,
      raise=5pt
    },
    decorate] (inprod.east) -- (inprod.west)
    node [pos=0.5,anchor=north,yshift=0cm] (a) {}; 
    \draw [thick,->,color=mixed,
    decorate, decoration={
      snake,
      segment length=4,
      amplitude=0.8,
      post=lineto,
      post length=2pt
    }]  (a) to [out=-90,in=90] ($(score.center)+(0.3,0.2)$);
    \draw [thick,->,color=report] (pp.center)+(5pt,-1pt) to [out=0,in=-90] ($(score.center)+(-0.1,-0.2)$);
  \end{tikzpicture}
  \caption{Information elicitation in classical and quantum settings.
    In the classical setting (left), the principal chooses a scoring function $s(p',y)$ that depends on the report $p'$ and observed outcome $y$.
    The outcome $y$ is drawn from the underlying distribution, which from the agent's point of view is their belief $p$.
    In the quantum setting (right), the principal again chooses a scoring function $s(\rho',y)$ which depends on the reported quantum mixed state $\rho'$ and outcome $y$, but also specifies a measurement function $\mu(\rho')$.
    The outcome $y$ is drawn from the distribution $\inprod{\mu(\rho')}{\rho}$, which is determined by both the chosen measurement and the underlying quantum mixed state, which from the agent's point of view is their belief $\rho$.}
  \label{fig:overview}
\end{figure}
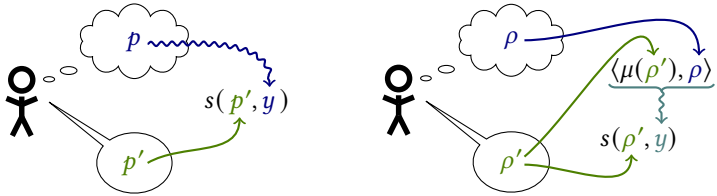

This paper introduces quantum information elicitation, presenting a model and several results to answer these and other questions.
In the basic model considered, a quantum scoring rule (or \emph{quantum score}), the principal chooses both a scoring function $s$ and a measurement function $\mu$, and the agent responds with their reported density matrix $\rho'$ (Fig.~\ref{fig:overview}).
The most general model allows the principal to choose the measurement as a function of the report $\rho'$.
Interestingly, this flexibility is not required to design a truthful quantum score:
if the principal fixes a \emph{tomographically complete} measurement, then the projection from the true mixed state $\rho$ onto a classical distribution described above is of full rank, and one can use a proper scoring rule to elicit the outcome distribution and therefore recover $\rho$.
Despite this straightforward reduction to the classical elicitation setting, several interesting characterization questions remain for particular classes of quantum scores; for example, the most common type of measurement, based on an orthonormal basis, is never tomographically complete.
It is also natural to ask about eliciting \emph{properties} (e.g.,\ summary statistics) of a density matrix, such as its eigenvectors or eigenvalues.

This work leverages the affine score framework of \citet{frongillo2021general} to give several characterization results, all of which echo a familiar structure based on convex analysis (\S~\ref{sec:quant-scor-rules}).
One interesting class of quantum scores are spectral scores, wherein the measurement function is given by the spectral decomposition of the report $\rho'$; spectral scores have particularly elegant structure and connections to quantum information theory (\S~\ref{sec:spectral-scores}).
Turning to property elicitation, we find that one can elicit eigenvectors of a density matrix, but not its eigenvalues or entropy (\S~\ref{sec:properties}).
In fact, the \emph{elicitation complexity} of eigenvalues and entropy, a measure of how hard it is to elicit them, is as large as that of the full density matrix itself (\S~\ref{sec:elic-complex}).
We conclude in \S~\ref{sec:discussion} with a discussion of other quantum information elciatition settings and connections to the literature.

\subsection{Prior work}

The most relevent prior work comes from the literature on the statistical learning of quantum states, e.g., practical tomography via machine learning \cite{granade2012robust,granade2016practical,qi2013quantum,shangnan2021quantum}.
We will discuss several connections to this literature in \S~\ref{sec:discussion}.
In preparing the manuscript, the author because aware of two such works which are especially relevant.
The first, \citet{blume2006accurate}, details a principal--agent setting for the elicitation of a quantum mixed state, and derives some of the same results about spectral scores which we discuss in \S~\ref{sec:spectral-scores}.
A follow-up work, \citet{blume2010optimal}, gives an informal argument for our main characterization (Theorem~\ref{thm:main-char}).
We will discuss the specific relation to these two works in context.

\section{Background}

\subsection{Convex analysis}
\label{sec:convex-analysis}

We will work with convex functions on both $\reals^n$ and $\C^n$; see \citet{rockafellar1997convex} and \citet{zalinescu2002convex} for general background.
Define the extended real line $\extreals := \reals \cup \{\infty,-\infty\}$.
For brevity, for a convex subset $C\subseteq\V$ of a real vector space $\V$, we write $G:C\to\reals$ to denote a (proper) convex function $G:\V\to\extreals$ with $\dom(G) = C$, i.e., such that $G(x) \in \reals$ for $x\in C$ and $G(x) = \infty$ for $x\notin C$.

For real vector spaces $\V_1,\V_2$, let $\lin(\V_1\to\V_2)$ denote the set of linear functions from $\V_1$ to $\V_2$.
Given a point $x\in C$, we say $d_x \in \lin(\V\to\reals)$ is a \emph{subgradient of $G$ at $x$} if for all $x'\in C$ we have $G(x') \geq G(x) + \inprod{d_x}{x'-x}$; here and throughout we denote the application of $d_x$ by an inner product.
The set of subgradients of $G$ at $x$ is denoted $\partial G(x)$.

While $\extreals$ does not form a real vector space, to capture scoring rules such as the log score (\S~\ref{sec:class-scor-rules}), we will need to define the set $\lin(\V \to \extreals)$ of ``extended'' linear functions.
The reader interested in these details and definitions can refer to \S~\ref{sec:app-ext-reals}.
We may define an \emph{extended subgradient} exactly as a subgradient but allowing $d_x \in \lin(\V\to\extreals)$.
We will denote the set of extended subgradients of $G$ at $x$ by $\extsub G(x)$.
In this context, a \emph{selection of subgradients} is a set $\{d_x\in\extsub G(x) \mid x \in C\}$ such that for all $x,z\in C$ we have $\inprod{d_x}{z-x} \in \reals\cup\{-\infty\}$.
We will write this condition more succinctly as \subselect{d}{x}{G}.

\subsection{Classical scoring rules}
\label{sec:class-scor-rules}

Perhaps the most fundamental information elicitation setting is a proper scoring rule.
Here a principal wishes to elicit an agent's probabilistic belief about a future outcome to be observed.

\begin{definition}
  Let $\Y$ be a finite set.
  Given $\P \subseteq \simplex$, a \emph{(classical) scoring rule} is a function $\hat s:\P\times\Y\to \extreals$.
  We write the expected score as $\hat s(q;p) := \E_{Y\sim p} \hat s(q,Y)$.
  We say $\hat s$ is \emph{proper} if we have $\hat s(q;p) \leq \hat s(p;p)$ for all $p,q \in \P$, and \emph{strictly proper} if this inequality is strict whenever $q \neq p$.
  Finally, $\hat s$ is \emph{regular} if for all $p,q\in\P$ we have $\hat s(q;p) \in \reals \cup \{-\infty\}$ and $\hat s(p;p)\in\reals$.
\end{definition}

Let $\ones_y \in\simplex$ denote the point mass on $y$, i.e., with $(\ones_y)_y = 1$ and $(\ones_y)_{y'} = 0$ for $y'\neq y$.
\begin{theorem}[\citet{frongillo2021general,gneiting2007strictly}]
  \label{thm:scoring-rule-char}
  Given $\P \subseteq \simplex$, a regular scoring rule $\hat s:\P\times\Y\to\extreals$ is (strictly) proper if and only if there exists a (strictly) convex function $G:\convhull(\P)\to\reals$ and selection of subgradients \subselect{dG}{p}{G} such that
  \begin{equation}
    \label{eqn:sr}
    \hat s(p,y) = G(p) + \inprod{dG_p}{\ones_y - p}~.
  \end{equation}
\end{theorem}

The two most common examples of proper scoring rules are Brier score and log score, given as follows,
\begin{equation}
  \label{eq:brier-classical}
  \brier(p,y) = 2p_y - \|p\|^2, \quad \brier(p';p) = 2\inprod{p'}{p} - \inprod{p'}{p'}~,
\end{equation}
\begin{equation}
  \label{eq:log-classical}
  s(p,y) = \log p_y, \quad s(p';p) = \sum_{y\in\Y} p_y \log p'_y~.
\end{equation}

We will also make use of a generalization of scoring rules, \emph{affine scores}~\citep{frongillo2021general}.
The affine score model treats scoring functions $s$ on some type space $\mathcal T\subseteq \V$ for a vector space $\V$, such that a suitably defined ``expected score'' $s(t';t)$ is affine as a function of the true type $t\in\mathcal T$.
The characterization in this more general setting resembles Theorem~\ref{thm:scoring-rule-char}; in particular, similar to eq.~\eqref{eqn:sr}, one can always write the score $s$ as a linear approximation of a convex function $G:\convhull(\mathcal T)\to\reals$ .

\subsection{Complex linear algebra}
\label{sec:linear-algebra}

Following \citet{watrous2018theory}, let $\X$ be a finite-dimensional complex Euclidean space, which we generically take to be $\X = \C^n$ for some $n \in \N$.
The identity operator is denoted $I_\X \in \lin(\X\to\X)$.
We write the conjugate transpose of $X\in \lin(\X\to\X)$ as $X^*$.
The inner product on $\lin(\X\to\X)$ is given by $\inprod{X}{Y} = \Tr(X^* Y)$, where $\Tr$ is the trace operator.

We say $U\in\lin(\X\to\X)$ is \emph{unitary} if $UU^* = U^*U = I_\X$.
We say $X\in \lin(\X\to\X)$ is \emph{Hermitian} if $X = X^*$, and denote by $\Herm(\X)\subseteq \lin(\X\to\X)$ the set of Hermitian operators on $\X$.
We will routinely use the fact that $\Herm(\X)$ is itself a vector space over the reals.
In \S~\ref{sec:extend-herm}, we also define the \emph{extended Hermitian matrices} $\extHerm(\X)$, which can be thought of as a particular subset of extended linear functions $\lin(\Herm(\X)\to\extreals)$.
Let $\Pos(\X)\subseteq\Herm(\X)$ be the set of Hermitian positive-semidefinite elements of $\lin(\X\to\X)$, i.e., those $X\in \lin(\X\to\X)$ which can be written $X=YY^*$ for some $Y\in \lin(\X\to\X)$.

Define $\lambda:\Herm(\X)\to\reals^n$ such that $\lambda(A)$ is the vector of eigenvalues of $A$ in decreasing order.
We will often use the shorthand $\lambda_i(\cdot) := \lambda(\cdot)_i$.
\begin{theorem}[{\cite[Cor.\ 1.4]{watrous2018theory}}]\label{thm:spectral}
  For any $A\in\Herm(\X)$, there exists an orthonormal basis $\{x_1,\ldots,x_n\}$ such that $A = \sum_{i=1}^n \lambda_i(A) \vvstar{x}{i}$.
\end{theorem}
One can verify that the $x_i$ are eigenvectors of $A$ with eigenvalue $\lambda_i(A)$.
We refer to the sum in Theorem~\ref{thm:spectral} as a \emph{spectral decomposition} of $A$.
This decomposition is not unique, for two reasons: (1) $A$ may have repeated eigenvalues, and (2) eigenvectors are only determined up to phase, meaning a multiplication by $e^{i\theta}$ for any $\theta \in\reals$.

The following basic properties of the inner product on Hermitian matrices will be used throughout.
The first relates the matrix inner product to the inner product of eigenvalues, showing in particular that once one fixes the eigenvalues, the highest inner product comes from choosing the same eigenbasis.
The second shows even more structure for positive-semidefinite matrices; an important corollary is that those with inner product 0 commute and can be simultaneously diagonalized.
\begin{lemma}[{\cite{neumann1937some,lewis1996convex}}]\label{lem:inner-prod-eigenvalues}
  For all $A,B\in\Herm(\X)$, we have $\inprod{A}{B} \leq \inprod{\lambda(A)}{\lambda(B)}$, with equality if and only if there exists a unitary matrix $U$ such that $A = U \Diag(\lambda(A)) U^*$ and $B = U \Diag(\lambda(B)) U^*$.
\end{lemma}

\begin{lemma}[\cite{hayashi2016quantum}]
  \label{lem:pos-zero-inner-prod}
  Let $A,B\in \Pos(\X)$.
  Then $\inprod{A}{B} \geq 0$ with equality if and only if $AB=0$.
\end{lemma}

\subsection{Quantum states, measurement, and entropy}
\label{sec:quant-stat-meas}

The analog of a classical probability distribution in quantum information theory is a quantum mixed state.
Mixed states are represented by a \emph{density matrix}, which is an element $\rho\in\Pos(\X)$ such that $\Tr(\rho) = 1$.
The set of all density matrices over $\X$ is denoted $\Dens(\X)$.
(The term \emph{mixed state} is in contrast to a \emph{pure state}, which is the case where $\rho$ has rank $1$; a mixed state can therefore be thought of as a convex combination of pure states.)

While classically one thinks of outcomes being drawn directly from a distribution, in the quantum setting, outcomes are drawn indirectly through the choice of a measurement.
In particular, the measurement induces a linear map from the density matrix to a classical distribution, from which the outcome is drawn (Fig.~\ref{fig:overview}).
Formally, a \emph{postive operator-valued measurement (POVM)}, or simply \emph{measurement}, is a labeled collection $\mu = \{\mu_y \in \Pos(\X)\}_{y\in\Y}$ for some finite set of outcomes $\Y$, such that $\sum_{y\in\Y} \mu_y = I_\X$.
If the measurement $\mu$ is applied to a mixed state $\rho\in\Dens(\X)$, the resulting distribution $p\in\simplex$ assigns probability $p_y = \inprod{\mu_y}{\rho}$ to outcome $y\in\Y$.
The condition $\mu_y \in \Pos(\X)$ ensures $p_y \geq 0$, and the condition $\sum_{y\in\Y} \mu_y = I_\X$ ensures $\sum_{y\in\Y} p_y = \inprod{\sum_{y\in\Y} \mu_y}{\rho} = \inprod{I_\X}{\rho} = 1$, thus justifying the claim $p\in\simplex$.
Given $\mu$ and $\rho$, we denote this distribution $p$ more succinctly as $p = \inprod{\mu}{\rho} \in \simplex$.
The set of measurements over $\X$ labeled by $\Y$ is denoted $\Meas_\Y(\X)$, with $\Meas(\X)$ denoting all possible measurements over $\X$ for any labeling.

An important class of measurements are \emph{projection-valued measurements (PVMs)}, which are measurements $\mu \in \Meas_\Y(\X)$ such that $\{\mu_y\}_{y\in\Y}$ is a set of orthogonal projection operators.
A fundamental special case, sometimes called a \emph{von Neumann measurement}, is when one picks an orthonormal basis $\{x_1,\ldots,x_n\}$ of $\X = \C^n$, and takes $\Y = \{1,\ldots,n\}$ and $\mu_y = \vvstar{x}{y}$.
In the literature on quantum information theory and quantum computation, one speaks of measuring in a particular basis---in other words, choosing a projective measurement in this way for a particular orthonormal basis---and some textbooks focus exclusively on this type of measurement~\cite{debrota2021varieties}.

\begin{example}[Basics of quantum states and measurements]
  Suppose $\rho$ were a mixture of two pure states, $\rho = \tfrac 1 3 \matmat{1/\sqrt 2;-1/\sqrt 2}\matmat{1/\sqrt 2 -1/\sqrt 2} + \tfrac 2 3 \matmat{0;1}\matmat{0 1} = \matmat{1/6 \,-1/6;1/6 5/6}$.
  Measuring in the standard basis $\matmat{1;0},\matmat{0;1}$ gives probability
  $\inprod{\matmat{1;0}\matmat{1 0}}{\rho} = \Tr(\matmat{1 0;0 0}\rho) = 1/6$ for outcome $1$ and $5/6$ for outcome $2$.
  Measuring instead in the Hadamard basis $\matmat{1/\sqrt 2; 1/\sqrt 2},\matmat{1/\sqrt 2;-1/\sqrt 2}$ gives probability
  $\inprod{\matmat{1/\sqrt 2;1/\sqrt 2}\matmat{1/\sqrt 2 1/\sqrt 2}}{\rho} = \Tr(\matmat{1/2 1/2;1/2 1/2}\rho) = 2/3$ for outcome $1$ and $1/3$ for outcome $2$.
\end{example}

A standard question in quantum information theory is when one can (approximately) infer the full mixed state $\rho\in\Dens(\X)$ from a sufficient number of independent repeated measurements.
To this end, we say a measurement $\mu \in \Meas_\Y(\X)$ is \emph{tomographically complete} if the linear span of $\{\mu_y\}_{y\in\Y}$ over the reals is $\Herm(\X)$.
Define $\phi:\Herm(\X)\to\reals^\Y$ to be the linear operator $\phi X = \inprod{\mu}{X}$, and observe that $\phi(\Dens(\X)) \subseteq \simplex$.
Then $\mu$ is tomographically complete if and only if the rows of $\phi$ span $\Herm(\X)$, if and only if $\phi$ has full rank.
In particular, if $\mu$ is tomographically complete, then $\phi$ is injective, and one can infer $\rho$ from the distribution over outcomes $\phi \rho = \inprod{\mu}{\rho}$ given by the measurement $\mu$.
Any tomographically complete measurement must have $|\Y| \geq n^2 = \dim(\X)^2$~\cite{debrota2021varieties}.

To quantify the amount of information contained in a quantum mixed state, a variety of entropy functions have been proposed.
Just as Shannon entropy $H(p) = -\sum_{y\in\Y} p_y \log p_y$ is the most commonly used for classical distributions $p\in\simplex$, so is von Neumann entropy for quantum mixed states:
\begin{equation}
  \label{eq:von-neumann-entropy}
  H(\rho) = -\inprod{\log \rho}{\rho}~,
\end{equation}
where $\log \rho = \sum_{i=1}^n \log \lambda_i(\rho) \vvstar x i$ for any spectral decomposition of $\rho$ from Theorem~\ref{thm:spectral}.
Von Neumann entropy is a generalization of Shannon: if $p\in\simplex$, then $H(\Diag(p)) = H(p)$.
Tsallis and R\'enyi entropy are also considered, alongside several others~\cite{hu2006generalized,rastegin2011some}.

\section{Quantum Scoring Rules}
\label{sec:quant-scor-rules}

The quantum information elicitation setting we consider is as follows.
An agent possesses some private belief in the form of a quantum mixed state $\rho$, and will report some potentially different mixed state $\rho'$ to the principal.
The principal commits to a contract $S = (s,\mu)$, called a \emph{quantum scoring rule}, or more succinctly \emph{quantum score}, which specifies a payment $s$ and measurement $\mu$.
\footnote{We will prefer the term ``quantum score'' to avoid confusion with the scoring function $s$, which in constrast to the classical case does not fully specify the contract.}
The measurement $\mu$ may depend on the report $\rho'$, and the payment $s$ may depend on both $\rho'$ and the outcome of the measurement (Fig.~\ref{fig:overview}).
The goal of the principal is to design the contract $S$ so that the agent maximizes their expected score by reporting truthfully, i.e., setting $\rho'=\rho$.

Formally defining a quantum score is somewhat subtle, since the measurement $\mu$ depends on the report $\rho'$, and therefore so does the space of possible outcomes $\Y(\rho')$.
Thus, the domain of the scoring function $s$ itself depends on the report $\rho'$.
To simplify notation, we will assume $\Y(\rho') \subseteq \N$, so that we may write $s:\Dens(\X)\times\N \to \extreals$;
we simply ignore values $s(\rho',y)$ for $y\notin\Y(\rho')$.

\begin{definition}
  A \emph{quantum score} is a pair $S=(s,\mu)$ of a scoring function $s:\Dens(\X)\times\N \to \extreals$ and measurement function $\mu:\Dens(\X)\to\Meas(\X)$.
  We say $S$ is \emph{finite} if $s$ takes values in $\reals$ only.
  The expected score is given by
  \begin{equation}
    \label{eq:quantum-def-expected-score}
    S(\rho';\rho) := \E_{Y \sim \inprod{\mu(\rho')}{\rho}} s(\rho',Y)~.
  \end{equation}
  A quantum score $S$ is \emph{truthful} if for all $\rho,\rho'\in\Dens(\X)$ we have $S(\rho;\rho) \geq S(\rho';\rho)$.
  We say $S$ is \emph{strictly truthful} if the above inequality is strict whenever $\rho'\neq\rho$.
\end{definition}

As we will see, in some sense this formulation of quantum scores is ``over-determined'', since there can be multiple ways to achieve the same expected score function with different choices of scoring and measurement functions.
The following notion of equivalence is therefore useful.
\begin{definition}\label{def:equivalent}
  Two quantum scores $S,S'$ are \emph{equivalent} if we have $S(\rho';\rho)=S'(\rho';\rho)$ for all $\rho,\rho'\in\Dens(\X)$.
  Furthermore, we say a particular class $\mathcal C$ of quantum scores \emph{can express} another class $\mathcal C'$ if for all quantum scores $S'\in\mathcal C'$ there is some $S\in\mathcal C$ which is equivalent to $S'$.
\end{definition}
\subsection{Motivating examples}

At first glance, it may not seem clear whether there exists any truthful quantum score.
Using tools from quantum tomography, however, we can see one construction via reduction to the classical case.
In particular, a fixed tomographically complete measurement suffices: simply score the induced classical distribution over outcomes using any strictly proper scoring rule.

\begin{example}[Fixed-measurement reduction to classical]
  \label{ex:classical-reduction}
  Let $\mufix\in\Meas_\Y(\X)$ be any tomographically complete measurement over a finite $\Y$.
  Then the operator $\phi:\Herm(\X)\to\reals^\Y$, $\phi X = \inprod{\mufix}{X}$ is injective.
  Now take any strictly proper scoring rule $\hat s:\simplex\times\Y\to\extreals$ and define our quantum score $S=(s,\mu)$ where $\mu:\rho\mapsto\mufix$ and
  $s(\rho,y) = \hat s(\phi \rho,y)$.
  Since $\hat s$ is strictly proper, we have
  \begin{equation}
    \label{eq:reduction-classical}
    S(\rho';\rho) =
    \E_{Y\sim \phi \rho} \hat s(\phi \rho',Y)
    \leq
    \E_{Y\sim \phi \rho} \hat s(\phi \rho,Y)
    = S(\rho;\rho)~,
  \end{equation}
  with equality if and only if $\phi \rho = \phi \rho' \iff \rho = \rho'$ by injectivity.
  As concrete instantiations, taking $\hat s$ to be log score yields $s(\rho',y) = \log \inprod{\mufix}{\rho'}_y$, and taking Brier scores yields $s(\rho',y) = 2\inprod{\mufix}{\rho'}_y - \sum_{y\in\Y}\inprod{\mufix}{\rho'}_y^2$, both of which yield strictly truthful quantum scores.
\end{example}  

The fact that the principal may simply use a classical proper scoring rule may seem to render the problem of designing truthful quantum scores trivial.
Indeed, the reader looking only for a straightforward family of truthul quantum scores may now be satisfied.
It is not clear, however, whether \emph{all} truthful quantum scores can be expressed as in Example~\ref{ex:classical-reduction}---in fact, the answer is not quite (Proposition~\ref{prop:projective-expressive} and preceeding discussion).
Moreover, we would like to design truthful quantum scores under various restrictions, such as requiring the measurement to be projective.

As an example, what if we are only allowed a binary-valued measurement?
In fact, we still have the flexibility to recover a quantum version of the classical Brier score.

\begin{example}[Brier via binary measurement]
  \label{ex:brier-binary-povm}
  We can also derive a quantum version of Brier score, using just two outcomes for the measurement.
  Take $\Y = \{0,1\}$ and define $\mu$ by $\mu(\rho')_0 = I_\X - \rho' \in \Pos(\X)$ and $\mu(\rho')_1 = \rho'$.
  Take the scoring function given by $s(\rho',y) = 2y - \inprod{\rho'}{\rho'}$.
  The expected score for $S=(s,\mu)$ is then
  \begin{align*}
    S(\rho';\rho)
    &= \inprod{\mu(\rho')_0}{\rho} s(\rho',0) + \inprod{\mu(\rho')_1}{\rho} s(\rho',1)
    \\
    &= \inprod{\rho'}{\rho} (2 - \inprod{\rho'}{\rho'}) + \inprod{I_\X - \rho'}{\rho} (- \inprod{\rho'}{\rho'})
    \\
    &= 2\inprod{\rho'}{\rho} - \inprod{\rho'}{\rho'}~,
  \end{align*}
  the same form as the classical Brier score (eq.~\eqref{eq:brier-classical}).
  As $2\inprod{\rho'}{\rho} - \inprod{\rho'}{\rho'} = \inprod{\rho}{\rho}-\inprod{\rho'-\rho}{\rho'-\rho}$,  the expected score is uniquely maximized by setting $\rho'=\rho$.
  Hence, $S$ is strictly truthful.
\end{example}

See 
Example~\ref{ex:brier-projective} for a projective version of the same score.
As we will see in Example~\ref{ex:von-neumann}, there is also a natural quantum analog of the log scoring rule.

\subsection{Main characterization}

We now characterize all truthful quantum scores.
To begin, let us examine the form of the expected score.
While it may not be clear from eq.~\eqref{eq:quantum-def-expected-score}, the expected score function for $S=(s,\mu)$ is actually (extended) linear in the belief $\rho$.
\begin{align*}
  S(\rho';\rho)
  &= \E_{Y \sim \inprod{\mu(\rho')}{\rho}} s(\rho',Y)
  \\
  &= \sum_{y\in\Y(\rho')} \inprod{\mu(\rho')_y}{\rho} s(\rho',y)~,
  \\
  &= \inprod{Z_S(\rho')}{\rho} ~,\numberthis\label{eq:coeff-rho-main-char}
\end{align*}
where $Z_S(\rho') = \sum_{y\in\Y(\rho')} \mu(\rho')_y s(\rho',y)$.
Technically, $Z_S(\rho') \in \extHerm(\X)$ since $s$ may take on infinite values.
If $S$ is finite, however, then $Z_S(\rho')\in\Herm(\X)$.

Because of the linear relationship in eq.~\eqref{eq:coeff-rho-main-char}, we may appeal to the affine score framework for a characterization.
We state the result as a theorem, but it follows as a direct corollary of \citet[Theorem 1]{frongillo2021general}.
An informal version of this result appears in \citet[eq. (34)]{blume2010optimal}.

\begin{theorem}\label{thm:main-char}
  A quantum score $S$ on $\X$ is (strictly) truthful if and only if
  there exists some (strictly) convex $F : \Dens(\X)\to\reals$, and some selection of subgradients \subselect{d}{\rho}{F},
  such that
  \begin{equation}
    \label{eq:main-char}
    S(\rho';\rho) = F(\rho') + \inprod{d_{\rho'}}{\rho-\rho'}~.
  \end{equation}
\end{theorem}

While it may not be clear from Theorem~\ref{thm:main-char}, the flexibility of choosing the measurement to be any POVM means $F$ can be any convex function.
See Corollary~\ref{cor:score-any-F-dF}.
This fact essentially follows from the observation that a quantum score implicitly specifies an arbitrary (extended) Hermitian matrix $Z_S(\rho')\in\extHerm(\X)$ so that $S(\rho';\rho) = \inprod{Z_S(\rho')}{\rho}$.

\subsection{Fixed-measurement quantum scores}
\label{sec:fixed-measurement-scores}

In light of Example~\ref{ex:classical-reduction}, an important special case is when a fixed measurement $\mufix\in\Meas_\Y(\X)$, possibly not tomographically complete, will be performed on the unknown mixed state $\rho$.
This case corresponds to a quantum score given by a constant measument function $\mu:\rho\mapsto \mufix$ and a scoring function $s:\Dens(\X)\times\Y\to\extreals$.

We first observe that the expected score function can only depend on the actual distribution over outcomes.
\begin{lemma}\label{lem:fixed-meas-expected-score}
  Let $S=(s,\mufix)$ be a truthful quantum score with fixed measurement $\mufix\in\Meas_\Y(\X)$.
  Let $F(\rho) := S(\rho;\rho)$.
  Then there exists $f:\simplex\to\reals$ such that $F(\rho) = f(\inprod{\mufix}{\rho})$.
\end{lemma}
\begin{proof}
  Suppose $\inprod{\mufix}{\rho} = \inprod{\mufix}{\tau} =: p \in \simplex$; we will show $F(\rho) = F(\tau)$.
  By definition of the expected score, we have $S(\rho;\rho) = \E_{Y\sim p} s(\rho,Y) = S(\rho;\tau)$, and similarly $S(\tau;\tau) = \E_{Y\sim p} s(\tau,Y) = S(\rho;\tau)$.
  By truthfulness when $\rho$ is the belief, $S(\rho;\rho) \geq S(\tau;\rho)$.
  By definition of the expected score, and the fact that $\inprod{\mufix}{\rho} = p$, we have $\E_{Y \sim p} s(\rho,Y) \geq \E_{Y \sim p} s(\tau,Y)$.
  Symmetrically, when $\tau$ is the belief, the reverse inequality holds since we also have $\inprod{\mufix}{\tau} = p$.
  We conclude $F(\rho) = \E_{Y \sim p} s(\rho,Y) = \E_{Y \sim p} s(\tau,Y) = F(\tau)$.
\end{proof}

Recall from Example~\ref{ex:classical-reduction} that one way to design a truthful quantum score is to use a fixed tomographically complete measurement and score the classical distribution on measurement outcomes induced by the report, using a classical proper scoring rule.
Using Lemma~\ref{lem:fixed-meas-expected-score}, we can show a converse: a truthful fixed-measurement quantum score can essentially be written in terms of a classical scoring rule.
There is a subtlety, however, which is that when $\mufix$ is not tomographically complete, the quantum score may be able to make multiple subgradient selections for the same induced classical distribution over outcomes.
In other words, we could have $d_{\rho}\neq d_{\rho'}$ below despite $\inprod{\mufix}{\rho} = \inprod{\mufix}{\rho'}$, meaning in this case we could not write eq.~\eqref{eq:fixed-meas-char} as a classical scoring rule.
When $\mufix$ is tomographically complete, however, eq.~\eqref{eq:fixed-meas-char} always gives a classical scoring rule.

\begin{restatable}{theorem}{fixedmeaschar}
  \restatehack{\begin{theorem}}
    \label{thm:fixed-meas-char}\restatehack{\end{theorem}}
  Let $S=(s,\mufix)$ be a quantum score with fixed measurement $\mufix\in\Meas_\Y(\X)$.
  Then $S$ is truthful if and only if it there exists some convex $f : \simplex\to\reals$ and selection of subgradients $\{d_\rho \in \extsub f(\inprod{\mufix}{\rho})\}_{\rho\in\Dens(\X)}$
  such that
  \begin{equation}
    \label{eq:fixed-meas-char}
    s(\rho',y) = f(\inprod{\mufix}{\rho'}) + \inprod{d_{\rho'}}{\ones_y - \inprod{\mufix}{\rho'}}~.
  \end{equation}
  Moreover, $S$ is strictly truthful if and only if $f$ is strictly convex and $\mufix$ is tomographically complete.
\end{restatable}

With this characterization, we can now see that fixed-measurement quantum scores are essentially fully expressive.
One caveat is that the proof relies on the scores being finite, meaning the scoring function cannot take on (negative) infinite values.
This restriction arises as classical scoring rules can only take on infinite values for predictions at the boundary of the probability simplex~\cite{waggoner2021linear}, yet the set the map $\phi:\rho \mapsto \inprod{\mufix}{\rho}$ cannot map the boundary of $\Dens(\X)$ onto the boundary of the simplex, as $\Dens(\X)$ is not a finitely generated convex set (Lemma~\ref{lem:fixed-meas-boundary}).
In particular, the analog of log score, the log spectral score (Example~\ref{ex:von-neumann}), cannot be expressed as a fixed-measurement score.

\begin{proposition}
  \label{prop:fixed-meas-expressive}
  For any tomographically complete measurement $\mufix\in\Meas_\Y(\X)$, the class of quantum scores with fixed measurement $\mufix$ can express the class of finite quantum scores.
\end{proposition}
\begin{proof}
  Let $S$ be a finite quantum score.
  Fix $\rho'\in\Dens(\X)$.
  By eq.~\ref{eq:coeff-rho-main-char}, we may write $S(\rho';\rho) = \inprod{Z_S(\rho')}{\rho}$ for some $Z_S(\rho')\in\Herm(\X)$.
  By definition of tomographically complete, we may write $d_s(\rho') = \sum_{y\in\Y} \alpha_y(\rho') \mufixy$ for some coefficients $\alpha_y(\rho')\in\reals$.
  We simply let $S^\diamond = (s,\rho\mapsto\mufix)$ where $s(\rho',y) = \alpha_y(\rho')$.
  Then for all $\rho\in\Dens(\X)$, we have $S^\diamond(\rho';\rho) = \sum_{y\in\Y} \inprod{\mufixy}{\rho} s(\rho',y) = \inprod{\sum_{y\in\Y} \alpha_y(\rho') \mufixy}{\rho} = \inprod{Z_S(\rho')}{\rho} = S(\rho';\rho)$.
\end{proof}

\subsection{Projective quantum scores}
\label{sec:proj-quant-scor}

An important special case of positive operator-valued measurements (POVMs) are \emph{projection-valued measurements (PVMs)}, which are given by a set of orthogonal projection operators.
Let us call a quantum score $(S,\mu)$ \emph{projective} if $\mu(\rho)$ is a PVM for all $\rho\in\Dens(\X)$.

\begin{example}[Projective Brier score]
  \label{ex:brier-projective}
  For all $\rho\in\Dens(\X)$, let $\rho = \sum_{y=1}^n \lambda(\rho)_y \vvstar{x}{y}$ be a spectral decomposition.
  Take $\mu(\rho) \in \Meas_\Y(\X)$ given by $\mu(\rho)_y = x_yx_y^*$, where $\Y = \{1,\ldots,n\}$; note $\mu(\rho)$ is a PVM for all $\rho$.
  Set $s(\rho,y) = 2\lambda(\rho)_y - \inprod{\rho}{\rho}$.
  Then we have $S(\rho';\rho) = \sum_{y\in\Y}\inprod{\mu(\rho')_y}{\rho}(2\lambda(\rho')_y - \inprod{\rho'}{\rho'}) = 2\sum_{y\in\Y}\inprod{\lambda(\rho')_y\mu(\rho')_y}{\rho} - \inprod{\rho'}{\rho'} = 2\inprod{\rho'}{\rho} - \inprod{\rho'}{\rho'}$.
  Interestingly, we could also write $s(\rho,y) = 2\lambda(\rho)_y - \|\lambda(\rho)\|_2^2$, which is exactly Brier score on the distribution $\lambda(\rho) \in \simplex$.
  In fact, this type of spectral reduction to classical scoring rules can be generalized; see \S~\ref{sec:spectral-scores}.
\end{example}

Since PVMs are much more restrictive than POVMs, and have at most $n=\dim(\X)$ elements as opposed to the $n^2$ needed to be tomographically complete, one might expect that projective quantum scores are less expressive.
In fact, they are just as expressive as general quantum scores.

\begin{proposition}\label{prop:projective-expressive}
  Projective quantum scores can express the class of all quantum scores.
\end{proposition}
\begin{proof}
  Let $S$ be a quantum score.
  Fix $\rho'\in\Dens(\X)$.
  Let $\Y = \{1,\ldots,n\}$.
  By Lemma~\ref{lem:extherm}, we have a set of orthonormal vectors $x_1,\ldots,x_n\in\X$, integer $k$, and coefficients $\lambda_1,\ldots,\lambda_n \geq 0$, such that $Z_S(\rho') = \sum_{y=1}^k \lambda_y \vvstar{x}{y} - \infty \sum_{y=k+1}^n \vvstar{x}{y}$.
  Let $\mu'(\rho') \in \Meas_\Y(\X)$ with $\mu'(\rho')_y = \vvstar{x}{y}$.
  Define $s'(\rho',y) = \lambda_y$ where $\lambda_y = -\infty$ for $y > k$.
  Then letting $S' = (s',\mu')$, for all $\rho\in\Dens(\X)$, we have
  \begin{equation*}
    S(\rho';\rho)
    = \inprod{Z_S(\rho')}{\rho}
    = \inprod{\sum_{y\in\Y} s'(\rho',y) \mu'(\rho')_y}{\rho}
    = S'(\rho';\rho)~.
  \end{equation*}
\end{proof}

The proof shows that $Z_S(\rho)$ can be an arbitrary element of $\extHerm(\X)$, giving the following.
\begin{corollary}\label{cor:score-any-F-dF}
  Let $F : \Dens(\X)\to\reals$ be convex, and let \subselect{d}{\rho}{F} be a selection of subgradients.
  Then there exists a truthful quantum score $S$ taking the form of eq.~\eqref{eq:main-char} for this $F$ and $\{d_\rho\}$.
\end{corollary}

\section{Spectral Scores}
\label{sec:spectral-scores}

We now turn to a particularly natural special case of projective quantum scores.
Recall that the eigenvalues of the report $\rho'$ are nonnegative and sum to one; as such, one may interpret the vector of eigenvalues $\lambda$ as a classical probability distribution.
Moreover, the eigenvectors of $\rho$, taken to be orthonormal vectors, can be interpreted as a basis with which to measure, i.e., a PVM.
It is therefore natural to try the following.
Take the report $\rho'$ and interpret it as a classical distribution $\lambda$ (its eigenvalues) together with a PVM (an orthonormal eigenbasis); perform the measurement to produce an outcome, and use a classical scoring rule on $\lambda$ and this outcome.
We call such a quantum score a \emph{spectral score}.

For which proper scoring rules, if any, is the corresponding spectral score truthful?
It is perhaps unclear whether this scheme could work for any proper scoring rule; for example, while given a measurement basis $\mu$, the optimal report for the eigenvalues is $\inprod{\mu}{\rho}$, perhaps the agent could achieve a better expected score by misreporting the basis.

In fact, we will show that any proper scoring rule gives rise to a truthful spectral score.
This of quantum scores gives another nontrivial reduction to the classical setting, and gives a decision-theoretic view on many quantum entropies.

\subsection{Definition}
\label{sec:spectral-definition}

Let us formally define spectral scores.
Throughout this section, let $\Y = \{1,\ldots,n\}$, so that $\X = \C^\Y$.

\begin{definition}
  Let $\muspec : \Dens(\X) \to \Meas_\Y(\X)$ be a measurement given by $\muspec(\rho)_y = x_yx_y^*$ where  $\{x_1,\ldots,x_n\}$ is an arbitrary orthonormal eigenbasis such that $\rho = \sum_{y\in\Y} \lambda(\rho)_y x_yx_y^*$.
  Given a classical scoring rule $\hat s : \simplex\times\Y\to\extreals$, the corresponding \emph{spectral score} $S[\hat s]$ is the quantum score $(s,\muspec)$ with $s(\rho,y) = \hat s(\lambda(\rho),y)$.
\end{definition}

Note that spectral scores only evaluate $\hat s$ on ordered distributions $p$, i.e., those with $p_1 \geq p_2 \geq \cdots \geq p_n$.
As such, without loss of generality, we may take $\hat s$ to be permutation-invariant, meaning for any $p\in\simplex$ and any permutation $\pi:\Y\to\Y$, we have $\hat s(p,y) = \hat s(\pi p, \pi_y)$, where $(\pi p)_y = p_{\pi_y}$.
Also, observe that while we have not fully specified $\muspec$, any choice of eigenbasis will yield the same expected score function $\rho \mapsto S(\rho';\rho)$; see Lemma~\ref{lem:spectral-form}.

In fact, we have already seen an example of a truthful spectral score: Example~\ref{ex:brier-projective} is the Brier spectral score $S[\brier]$.
As another example, consider the following generalization of the log score.
\begin{example}[log spectral score]
  \label{ex:von-neumann}
  Let $S=S[\hat s]$ where $\hat s(\lambda,y) = \log \lambda_y$ is the log score.
  Computing the expected score, we have
  \begin{equation}
    \label{eq:log-example-expected-score}
    S(\rho';\rho) = \sum_{y\in\Y} \inprod{\vvstar{x}{y}}{\rho} \log \lambda(\rho')_y = \inprod{\log \rho'}{\rho}~,
  \end{equation}
  where $\lambda(\rho')$ and $\{x_1,\ldots,x_n\}$ is a spectral decomposition of $\rho'$.
  \begin{equation}
    \label{eq:log-example-divergence}
    S(\rho;\rho) - S(\rho';\rho) = \inprod{\log \rho - \log \rho'}{\rho} = D(\rho\|\rho')~,
  \end{equation}
  von Neumann relative entropy.
  Thus, since $D(\rho\|\rho')$ attains its minimum value of 0 if and only if $\rho = \rho'$, the score is truthful.
  Moreover, while the expected score of the log score $\hat s(\lambda;\lambda)$ is the Shannon entropy of $\lambda$, the expected score of its spectral quantum score $S(\rho;\rho)$ is the von Neumann entropy of $\rho$.
  The truthfulness of the log spectral score was also observed by \citet{blume2006accurate}.
\end{example}

While we have not fully specified $\muspec$, in particular because repeated eigenvalues may yield multiple choices for the orthonormal eigenbasis, all valid choices lead to the same expected score.
\begin{lemma}
  \label{lem:spectral-form}
  Let spectral score $S = S[\hat s]$ for a permutation-invariant classical scoring rule $\hat s$.
  Then for any $\rho\in\Dens(\X)$, and any spectral decomposition $\rho = \sum_{y\in\Y} \lambda(\rho)_y \vvstar{x}{y}$ of $\rho$, we have $Z_S(\rho) = \sum_{y\in\Y} \hat s(\lambda(\rho),y) \vvstar{x}{y}$.
  In particular, we have $S(\rho;\rho) = \hat s(\lambda(\rho);\lambda(\rho))$.
\end{lemma}

\subsection{Truthfulness}
\label{sec:spectral-truthfulness}

We now give the main result of this section: spectral scores are (strictly) truthful if and only if they are given by a (strictly) proper scoring rule.
The proof relies on elegant results from convex analysis about spectral convex functions of Hermitian matrices.
Specifically, the functions $F:\Herm(\X)\to\extreals$ which are convex and invariant under unitary transformations are exactly the functions of the form $F(\rho) = f(\lambda(\rho))$ where $f:\reals^n\to\extreals$ is convex.

Moreover, the subgradients of $F$ exactly correspond to the measurement $\muspec$, as we now show.
The proof is a careful generalization of \cite[Theorem 3.1]{lewis1996convex} to handle the extended linear case.
Specifically, we need to leverage the fact that in the classical case, subgradients can only have negative infinite entries when the corresponding probability is zero; analogously in the quantum case, negative infinite subgradient eigenvalues occur only when the corresponding density eigenvalue is zero.

\begin{restatable}{lemma}{lemspecsubgrad}
  \restatehack{\begin{lemma}}
    \label{lem:spectral-convex-subgradients}\restatehack{\end{lemma}}
  Let $f:\simplex\to\reals$ be permutation-invariant and convex, and let $F:\Dens(\X)\to\reals$ be given by $F(\rho) = f(\lambda(\rho))$.
  For any $\lambda\in\simplex$ and unitary matrix $U\in\lin(\X\to\X)$, we have
  $d\in\extsub f(\lambda)$ if and only if $U \Diag(d) U^* \in \extsub F(U \Diag(\lambda) U^*)$.
\end{restatable}
\begin{restatable}{lemma}{lemFfconvex}
  \restatehack{\begin{lemma}}
    \label{lem:F-f-strict-convex}\restatehack{\end{lemma}}
  Let $f:\simplex\to\reals$ be permutation-invariant, and let $F:\Dens(\X)\to\reals$ be given by $F(\rho) = f(\lambda(\rho))$.
  Then $F$ is (strictly) convex if and only if $f$ is (strictly) convex.
\end{restatable}
\begin{theorem}
  \label{thm:spectral-main-char}
  Let $S$ be the spectral quantum score given by a permutation-invariant classical scoring rule $\hat s$.
  Then $S$ is (strictly) truthful if and only if $\hat s$ is (strictly) proper.
\end{theorem}
\begin{proof}
  If $S$ is (strictly) truthful, it remains so when restricting to diagonal density matrices.
  For any $\lambda,\tau\in\simplex$, we have
  \begin{align*}
    S(\Diag(\lambda);\Diag(\tau))
    = \inprod{\hat s(\Diag(\lambda))}{\Diag(\tau)}
    = \inprod{\Diag(\hat s(\lambda,\cdot))}{\Diag(\tau)}
    = \hat s(\lambda;\tau)~.
  \end{align*}
  We conclude that $\hat s$ must be (strictly) proper.

  For the converse, suppose $\hat s$ is (strictly) proper.
  By Theorem~\ref{thm:scoring-rule-char}, we may write $\hat s(\lambda,y) = f(\lambda) + df_\lambda (\ones_y - \lambda)$, where $f:\simplex\to\reals$, $f:\lambda\mapsto \hat s(\lambda;\lambda)$ is (strictly) convex, and $df_\lambda \in\extsub f(\lambda)$ for all $\lambda\in\simplex$.
  Define $F(\rho) := S(\rho;\rho)$.
  By Lemma~\ref{lem:spectral-form}, we have $F(\rho) = \hat s(\lambda(\rho);\lambda(\rho)) = f(\lambda(\rho))$.
  By Lemma~\ref{lem:F-f-strict-convex}, $F$ is (strictly) convex.
  For any $\rho\in\Dens(\X)$, let $\{x_1,\ldots,x_n\}$ be its orthonormal eigenbasis chosen by $\muspec(\rho)$, i.e., so that $\muspec(\rho)_y = \vvstar{x}{y}$.
  Take $dF_\rho = U(\rho)^* \Diag(df_{\lambda(\rho)}) U(\rho) \in \extHerm(\X)$,
  where $U(\rho)$ is the unitary matrix with columns $x_1,\ldots,x_n$.
  Since $\rho = U(\rho)^* \Diag(\lambda(\rho)) U(\rho)$, Lemma~\ref{lem:spectral-convex-subgradients} gives $dF_\rho \in \extsub F(\rho)$.

  Finally, fix any $\rho,\rho'\in\Dens(\X)$, and let $U=U(\rho')$.
  By the proof of Lemma~\ref{lem:spectral-convex-subgradients}, we have $\inprod{dF_{\rho'}}{\rho'} = \inprod{df_{\lambda(\rho')}}{\lambda(\rho')}$.
  Thus,
  \begin{align*}
    S(\rho';\rho)
    &= \inprod{Z_S(\rho')}{\rho}
    \\
    &= \inprod{U^* \Diag(\hat s(\lambda(\rho'),\cdot)) U}{\rho}
    \\
    &= \inprod{U^* \Bigl(\Diag(f(\lambda(\rho'))\ones + df_{\lambda(\rho')} - \inprod{df_{\lambda(\rho')}}{\lambda(\rho')}\ones)\Bigr) U}{\rho}
    \\
    &= f(\lambda(\rho')) + \inprod{U^* \Diag(df_\lambda) U}{\rho} - \inprod{df_{\lambda(\rho')}}{\lambda(\rho')}
    \\
    &= F(\rho') + \inprod{dF_{\rho'}}{\rho} - \inprod{dF_{\rho'}}{\rho'}
      ~.
  \end{align*}
  By Theorem~\ref{thm:main-char}, the spectral score $S$ is (strictly) truthful.
\end{proof}

Most entropy functions in quantum information theory are functions of the eigenvalues alone, and typically strictly concave~\cite{hu2006generalized}.
All such entropy functions therefore give rise to strictly truthful spectral score, via the classical analog of the entropy function.

\citet{blume2006accurate} give a similar result to Theorem~\ref{thm:spectral-main-char} for a variation of spectral scores where the scoring function $s$ must be a proper scoring rule, but the measument $\mu$ need not be $\muspec$.
They find (their Lemma 1) that such a quantum score is truthful if and only if it is spectral, i.e., if and only if $\mu = \muspec$.
They further show that the only spectral score such that $s(\rho,y)$ depends only on $\lambda(\rho)_y$ is the log spectral score, following analogous results in the classical case~\cite{bernardo1979expected,aczel1980mixed,parry2012proper}.

\subsection{Characterization in terms of unitary invariance}
\label{sec:spectral-unitary-invariance}

Spectral scores are clearly a subset of all quantum scores.
The following characterizes this subset: spectral scores are exactly those which are unitary-invariant.
we say $S$ is \emph{unitary-invariant} if for all $\rho,\rho'\in\Dens(\X)$ and all unitary $U\in\Pos(\X)$ we have $S(\rho';\rho) = S(U\rho' U^*; U\rho U^*)$.

\begin{theorem}
  A finite truthful quantum score is equivalent to a spectral score if and only if it is unitary-invariant.
\end{theorem}
\begin{proof}
  Let $S$ be a given truthful quantum score.
  If $S$ is equivalent to a spectral score $S[\hat s]$, then it is unitary-invariant by Lemma~\ref{lem:spectral-form}:
  $S(U\rho' U^*; U\rho U^*) = \inprod{U Z_S(\rho') U^*}{U\rho U^*} = \inprod{Z_S(\rho')}{\rho} = S(\rho';\rho)$.

  Conversely, if $S$ is unitary-invariant, then $F(\rho) := S(\rho;\rho)$ is a unitary-invariant convex function.
  From~\cite{davis1957all,lewis1996convex}, we may therefore write $F(\rho) = f(\lambda(\rho))$ for $f:\simplex\to\reals$ convex.
  Moreover, from Theorem~\ref{thm:main-char} we may write $S(\rho';\rho) = F(\rho') + \inprod{dF_{\rho'}}{\rho-\rho'}$ for some selection of subgradients $dF$ of $F$.
  By Lemma~\ref{lem:spectral-convex-subgradients}, for each $\rho'\in\Dens(\X)$ we may write $dF_{\rho'} = U^*\Diag(df_{\rho'}) U$, where $df_{\rho'} \in \partial f(\lambda(\rho'))$ and $U$ is a unitary matrix that diagonalizes $\rho'$.
  In principle, these choices $df_{\rho'}$ may depend on $\rho'$; we will show that they depend only on $\lambda(\rho')$.

  Let $\lambda' \in \simplex$ and define $df_{\lambda'} = df_{\Diag(\lambda')}$.
  Let $\rho' = U^* \Diag(\lambda') U$ for a unitary matrix $U$.
  Let $\rho\in\Dens(\X)$ be arbitrary.
  Then
  \begin{align*}
    S(\Diag(\lambda');\rho)
    &= f(\lambda') + \inprod{\Diag(df_{\lambda'})}{\rho - \Diag(\lambda')}
    \\
    S(\rho';U^*\rho U)
    &= f(\lambda') + \inprod{U^*\Diag(df_{\rho'}) U}{U^* \rho U-\rho'}
    \\
    &= f(\lambda') + \inprod{\Diag(df_{\rho'})}{\rho-\Diag(\lambda')}~.
  \end{align*}
  Thus, by unitary-invariance, we have $\inprod{\Diag(df_{\rho'})}{\rho-\Diag(\lambda')} = \inprod{\Diag(df_{\lambda'})}{\rho - \Diag(\lambda')}$ for all $\rho\in\Dens(\X)$.
  Without loss of generality, then, we may assume $d_{\rho'} = df_{\lambda'}$, since both induce the same expected score function.
  In particular, we have
  \begin{align*}
    S(\rho';\rho)
    &= f(\lambda') + \inprod{U^* \Diag(df_{\lambda'}) U}{\rho-\rho'}
    \\
    &= \inprod{U^* \Diag(f(\lambda')\ones + df_{\lambda'} - \inprod{df_{\lambda'}}{\lambda'}\ones) U}{\rho}
    \\
    &= S[\hat s](\rho';\rho)~,
  \end{align*}
  where $\hat s(\lambda,y) = f(\lambda) + \inprod{df_\lambda}{\ones_y - \lambda}$.
\end{proof}

\section{Quantum Properties and Elicitability}
\label{sec:properties}

Density matrices $\rho$ on $\X = \C^n$ require $n^2-1$ real parameters to specify.
In some settings, however, the principal may only be interested in a particular statistic or \emph{property} of $\rho$, such as the maximum eigenvalue $\lambda_1(\rho)$ or von Neumann entropy $H(\rho)$.
When can we design scoring rules to incentivize the truthful reporting of such properties, i.e., which properties of density matrices are elicitable?
We apply tools from the classical case to study a few natural properties in the quantum setting.
For those that are not elicitable, in \S~\ref{sec:elic-complex} we further ask their \emph{elicitation complexity}: how many dimensions an elicitable property needs to have before one can recover the property of interest.

Informally, a property determines the ``correct'' reports, from some set $\R$, for a given density matrix $\rho$.
To score such reports, we will naturally need to generalize quantum scores, so the scoring function and measurement function depend on a report $r\in\R$ instead of some reported density $\rho'$.
\begin{definition}
  A \emph{quantum property} is any function $\Gamma : \Dens(\X) \to \R$ for some set $\R$.
  In this context, a \emph{quantum score} is a pair $S=(s,\mu)$ of a scoring function $s:\R\times\N \to \reals$ and measurement function $\mu:\R\to\Meas(\X)$.
  The expected score is given by
  \begin{equation}
    \label{eq:quantum-general-def-expected-score}
    S(r;\rho) := \E_{Y \sim \inprod{\mu(r)}{\rho}} s(r,Y)~.
  \end{equation}
  We say a quantum score $S$ \emph{elicits} $\Gamma$ if for all $\rho \in \Dens(\X)$,
  \begin{equation}
    \label{eq:quantum-property-elicits}
    \{\Gamma(\rho)\} = \argmax_{r\in\R} \; S(r;\rho)~.
  \end{equation}
  The \emph{level sets} of $\Gamma$ are the sets $\Gamma_r = \{\rho\in\Dens(\X) \mid \Gamma(\rho)=r\}$ for each $r\in\R$.
\end{definition}
We will occasionally deal with set-valued properties as well, $\Gamma : \Dens(\X) \to 2^\R$, which specify a set of correct reports for each density matrix.
Here we say $S$ elicits $\Gamma$ if we have $\Gamma(\rho) = \argmax_{r\in\R} \; S(r;\rho)$, i.e., the optimal reports for $\rho$ must be exactly $\Gamma(\rho)$.
In both cases, analogous to $Z_S(\rho)$, we may define $Z_S(r) = \sum_{y\in\Y(r)} \mu(r)_y s(r,y) \in \Herm(\X)$.

Translating an observation of Osband~\cite{osband1985providing}, the level sets of any elicitable property must be convex \cite[Corollary 4.9]{frongillo2019general}.
\begin{proposition}
  If $\Gamma:\Dens(\X)\to\R$ is elicitable, the level sets $\Gamma_r$ are convex for all $r\in\R$.
\end{proposition}

We now consider various properties of interest in the quantum setting, either deriving quantum scores to elicit them, or more often, showing that they are not elicitable.

\paragraph{Expectations and linear properties}
Notions of random variables and expected values are somewhat subtle for quantum states~\cite{hudson2018short}.
In essence, we can think of a random variable $Z$ as a pair $(z:\Y\to\reals^k,\mufix\in\Meas_\Y(\X))$ for some (for our purposes, finite) set $\Y$.
Then the expected value $\E_\rho Z$ under $\rho\in\Dens(\X)$ is simply $\E_{Y\sim\inprod{\mufix}{\rho}} z(Y)$, consistent with our usage in this article.
Since classical expectations are elicitable, so are quantum expectations.
Take any scoring rule $\hat s:\reals^k\times\Y\to\reals$ for the mean, and apply the fixed-measurement quantum score $S = (\hat s,\mufix)$: we have $S(r;\rho) = \hat s(r;\inprod{\mufix}{\rho})$, which is uniquely maximized by $r = \E_{Y\sim\inprod{\mufix}{\rho}} z(Y) = \E_\rho Z$.
The above generates scores of the form $S(r;\rho) = G(r) + \inprod{dG_r}{\E_\rho Z - r} + \E_\rho Z'$, where $G$ is convex with subgradients $dG$, and $Z' = (z',\mufix)$ is another quantum random variable~\cite{frongillo2015vector-valued}.
We may write this expectation as the property $\Gamma(\rho) = \E_\rho Z = \inprod{\sum_{y\in\Y} z(y) \mufixy}{\rho}$.
Since the coefficients $z(y)$ are arbitrary, we can see that the first term in the inner product ranges over all possible $A\in\Herm(\X)$.
Expected values thus form all possible linear maps from $\Dens(\X)$ to $\reals^k$.

\paragraph{Entropy, norms, entanglement}
Just as in the classical case, most entropy functions are strictly concave, and their level sets are typically non-convex.
In particular, one can check the non-convexity of level sets for all the examples in \S~\ref{sec:quant-stat-meas}, e.g., von Neumann, Tsallis, and  R\'enyi.
These functions are therefore all non-elicitable.
As many entanglement measures reduce to an entropy function as a special case, most are likely non-elicitable as well~\citep{horodecki2001entanglement,plenio2014introduction}.
Finally, essentially all non-trivial norms have non-convex level sets, such as the Schatten norms $\|\rho\|_\beta := \|\lambda(\rho)\|_\beta$, where the latter is the usual vector norm $\|x\|_\beta = (\sum_i |x_i|^\beta)^{1/\beta}$.

\paragraph{Eigenvalues}
An interesting query unique to the quantum setting is the vector of eigenvalues of a density matrix.
As we have seen, in many ways one can consider $\lambda(\rho)\in\simplex$ to be the classical distribution ``living inside'' the density matrix $\rho$.
As classical distributions are elicitable (via strictly proper scoring rules), and density matrices are (via strictly truthful quantum scores), one may guess that the function $\lambda : \Dens(\X) \to \simplex$, considered as a property, is also elicitable.
Unfortunately, $\lambda$ is not elicitable; again one can easily check that its level sets are non-convex.
(As a simple example, consider $\rho_1 = \tfrac 1 4 \matmat{1 0; 0 3}$ and $\rho_2 = \tfrac 1 4 \matmat{3 0; 0 1}$; then $\lambda(\rho_1)=\lambda(\rho_2)=(\tfrac 3 4, \tfrac 1 4)$, yet $\lambda(\tfrac 1 2 \rho_1 + \tfrac 1 2 \rho_2) = (\tfrac 1 2, \tfrac 1 2)$.)
In fact, even specifying a particular entry of $\lambda$, such as the maximum eigenvalue $\lambda_1$, leads to non-convex level sets.
(Same example.)
As such, eigenvalues are also non-elicitable.

\paragraph{Eigenvectors}
While eigenvalues of $\rho$ are not elicitable, eigenvectors are.
Specifically, one can elicit a set of orthonormal eigenvectors corresponding to the top (or bottom) $k$ eigenvalues, for any $k$.
Let us first see how to elicit an eigenvector of $\rho$ with the maximum eigenvalue.
Formally, set $\nsphere = \{x\in\X \mid \inprod{x}{x}=1\}$, and define $\Gamma_1 : \Dens(\X) \to 2^\nsphere$ to be the set of eigenvectors of $\rho$ with eigenvalue $\lambda_1(\rho)$.
Fix $\Y = \{0,1\}$ and take $S=(s,\mu)$ where $s(x,y) = \ones\{y = 1\}$ and $\mu(x) = \{I_\X - xx^*,xx^*\}$.
(Note that $s$ does not actualy depend on $x$.)
Then expected score is $S(x ; \rho) = \inprod{xx^*}{\rho} = \inprod{x}{\rho x}$, which is maximized when $x$ is an eigenvector with eigenvalue $\lambda_1(\rho)$.
(See \S~\ref{sec:discussion-finite-properties} for how $\Gamma_1$ relates to the mode for classical distributions.)

We can generalize this argument to the eigenvectors corresponding to the highest $k$ eigenvalues.
Let $\Gamma_{1..k}$ be the set-valued property specifying $k$ orthonormal vectors $x_1,\ldots, x_k$ with eigenvalues $\lambda_1(\rho),\ldots,\lambda_k(\rho)$.
Now set $\Y = \{1,\ldots,k,k+1\}$ and take $S=(s,\mu)$ as follows.
Let $s((x_1,\ldots,x_k),y) = v_y$ for some score vector $v\in\reals^n$, with $v_1 > v_2 > \ldots > v_k > v_{k+1} = \ldots = v_n = 0$.
(Note again that $s$ does not depend on the report.)
Let $\mu((x_1,\ldots,x_k)) = \{x_1x_1^*, \ldots, x_kx_k^*, I_\X - \sum_{i=1}^k x_ix_i^*\}$.
Then $S((x_1,\ldots,x_k);\rho) = \inprod{\sum_{i=1}^k v_i x_ix_i^*}{\rho}$.
Let $A = \sum_{i=1}^k v_i x_ix_i^*$ and note $\lambda(A) = v$.
By Lemma~\ref{lem:inner-prod-eigenvalues}, we have $S((x_1,\ldots,x_k);\rho) \leq \inprod{v}{\lambda(\rho)}$, with equality if and only if there exists a unitary $U$ such that $A = U \Diag(v) U^*$ and $\rho = U \Diag(\lambda(\rho)) U^*$.
The latter is equivalent to $x_1,\ldots,x_k$ being eigenvectors of $\rho$ with eigenvalues $\lambda_1(\rho),\ldots,\lambda_k(\rho)$.
In fact, this argument can even be generalized to simultaneously elicit the top $k$ eigenvectors and bottom $m$ eigenvectors.
\footnote{Require $v_1 > v_2 > \cdots > v_k > v_{k+1} = \cdots = v_{n-m} = 0 > v_{n-m-1} > \cdots > v_n$.  Set $s((x_1,\ldots,x_{k},x_{n-m-1},\ldots,x_n),y) = v_y$.
Let $\mu((x_1,\ldots,x_{k},x_{n-m-1},\ldots,x_n)) = \{x_1x_1^*, \ldots, x_nx_n^*\}$, where $x_{k+1},\ldots,x_{n-m}$ are chosen arbitrarily to complete the given orthonormal vectors to an orthonormal basis.  Then $S((x_1,\ldots,x_k,x_{n-m-1},\ldots,x_n);\rho) = \inprod{\sum_{i=1}^n v_i x_ix_i^*}{\rho}$, which is again upper bounded by $\inprod{v}{\lambda(\rho}$, again with equality if and only if the $x_i$ are eigvectors of $\rho$ with eigenvalues $\lambda_i(\rho)$.}

\section{Elicitation Complexity}
\label{sec:elic-complex}

When confronted with a non-elicitable property $\Gamma$, it is natural to ask its \emph{elicitation complexity}~\cite{lambert2008eliciting,frongillo2021elicitation}.
That is, what is the smallest dimension of an elicitable property $\hat \Gamma$ from which one can compute $\Gamma$?
In this section, we develop general tools which allow us to study the elicitation complexity of the non-elicitable properties from the previous section, namely entropy, norms, and eigenvalues.

\subsection{Definitions}

To avoid highly discontinuous constructions, one must place restrictions on the intermediate property $\hat\Gamma$ above.
Here, as in previous works, we require $\hat\Gamma$ to be \emph{indentifiable}, meaning each of its level sets $\{\rho \mid \Gamma(\rho) = r\}$ is defined by a linear constraint.
\begin{definition}
  \label{def:identifiable}
  A (quantum) property $\Gamma:\Dens(\X)\to\reals^k$ is \emph{identifiable} if, for each $r\in\Gamma(\Dens(\X))$, there exists $V(r)\in\Herm(\X)^k$ such that for all $\rho\in\Dens(\X)$,
  \begin{equation}
    \label{eq:iden}
    \Gamma(\rho)=r \iff \inprod{V(r)}{\rho} = 0 \in \reals^k~,
  \end{equation}
  where $\inprod{V(r)}{\rho} = (\inprod{V(r)_i}{\rho})_{i=1}^k$.
  In this case we call $V$ the \emph{identification function} for $\Gamma$.
\end{definition}
The analagous definition for the classical case takes $v:\R\to\reals^{k\times\Y}$ and requires $\Gamma(p) = r \iff \inprod{v(r)}{p} = 0 \in \reals^k$.
Identifiability is relatively a weak condition; for example, in the classical setting, any convex-elicitable property (one elicited by a convex loss function / concave scoring rule) is also identifiable~\cite{finocchiaro2021unifying}.

\begin{definition}
  \label{def:elic-complex}
  A (quantum) property $\Gamma:\Dens(\X)\to\R$ is \emph{$k$-elicitable}
  if there exists an elicitable and identifiable property $\hat \Gamma:\Dens(\X)\to\reals^k$ and map $\psi:\reals^k\to\R$ such that $\Gamma = \psi \circ \hat \Gamma$.
  The \emph{elicitation complexity} of $\Gamma$ is $\elic_Q(\Gamma) = \min\{k: \Gamma \text{ is $k$-elicitable}\}$.
\end{definition}

The analogous definition of elicitation complexity for the classical case follows using the corresponding definition of identifiability.
We will additionally restrict to a subset of classical distributions $\P \subseteq \simplex$:
for a (classical) property $\Gamma:\P\to\R$, we denote by $\elic_\P(\Gamma)$ its elicitation complexity when restricted to $\P$.

Immediate from Definition~\ref{def:elic-complex} is that if $\Gamma$ is a function of $\Gamma'$, it must have weakly lower elicitation complexity.
In this case we say $\Gamma'$ \emph{refines} $\Gamma$.

\begin{lemma}
  Let $\Gamma:\Dens(\X)\to\R,\Gamma':\Dens(\X)\to\R'$ such that $\Gamma = f \circ \Gamma'$ for some $f:\R'\to\R$.
  Then $\elicq(\Gamma) \leq \elicq(\Gamma')$.
\end{lemma}

\subsection{Reduction to the classical case}
\label{sec:reduct-class-case}

To understand quantum elicitation complexity, we will apply results for the classical case.
To bridge the two definitions, we will leverage and expand on the correspondance between fixed-measurement quantum scores and classical scoring rules from Example~\ref{ex:classical-reduction} and Proposition~\ref{prop:fixed-meas-expressive}.

For the remainder of this section, fix any tomographically complete $\mufix \in \Meas_\Y(\X)$ for some finite set $\Y$.
As before, define $\phi:\Herm(\X)\to\reals^\Y$ to be the linear operator $\phi X = \inprod{\mufix}{X}$.
As $\mufix$ is tomographically complete, the rows of $\phi$ span $\Herm(\X)$.
Thus $\phi$ has a linear left inverse, namely the Moore--Penrose pseudoinverse $\phi^+:\reals^\Y\to\Herm(\X)$.
The adjoint of $\phi^+$ is the linear operator $\phi^{+*}:\Herm(\X)\to\reals^\Y$.
Define $\Pfix = \{\inprod{\mufix}{\rho} \mid \rho \in \Dens(\X)\} \subseteq \simplex$, which is the image of $\Dens(\X)$ under $\phi$.
In what follows, given any property $\Gamma:\Dens(\X)\to\R$, we define $\Gammafix:\Pfix\to\R$ by $\Gammafix(p) = \Gamma(\phi^+ p)$.

\begin{restatable}{proposition}{propeliciden}
  \restatehack{\begin{proposition}}
    \label{prop:quantum-classical-elic-iden}\restatehack{\end{proposition}}
  Let $\Gamma:\Dens(\X)\to\R$.
  $\Gamma$ is elicitable if and only if $\Gammafix$ is elicitable, and $\Gamma$ is identifiable if and only if $\Gammafix$ is identifiable.
\end{restatable}

\begin{corollary}\label{cor:elic-complex-equiv}
  For all $\Gamma:\Dens(\X)\to\R$, we have $\elic_Q(\Gamma) = \elic_{\Pfix}(\Gammafix)$.
\end{corollary}

With the above equivalence, we can reduce the problem of quantum elicitation complexity to the classical elicitation complexity with respect to the subset $\Pfix \subseteq \simplex$.
This reduction suffices to establish the complexities of several properties relevant to quantum information theory.

\subsection{Entropy, norms, and eigenvalues}
\label{sec:complexity-entropy-norms-eigenvalues}

Classically, most entropy functions and norms are hard to elicit: they have the same complexity as eliciting the entire distribution itself.
This result follows from a more general one that any strictly convex (or strictly concave) function of the distribution has maximal elicitation complexity.
Using Corollary~\ref{cor:elic-complex-equiv}, we will show that the same holds for the quantum setting.

To begin, let us restate the result for the classical case.
Let $\id:\Dens(\X)\to\Dens(\X)$ be the identity, $\id: \rho\mapsto\rho$.
Then $\idfix:\Pfix\to\Pfix$ is also the identity, $\idfix:p\mapsto p$.
As $\Pfix$ is a convex subset of $\simplex$, we have the following from \citet[Corollary 2]{frongillo2021elicitation}.

\begin{proposition}
  \label{prop:classical-entropy-norm}
  Let $G:\Pfix\to\reals$ be strictly convex.
  Then $\elicfix(G) = \elicfix(\idfix)$.
\end{proposition}

Recall that we set $\X = \C^n$.
It is well-known that a density matrix on $\C^n$ requires $n^2-1$ real parameters to specify, which suggests that we should have $\elicfix(\idfix) = n^2-1$.
To see that this complexity is correct, note that as the image of $\Dens(\X)$ an injective linear map, the affine dimension of $\Pfix$ must be the affine dimension of $\Dens(\X))$, which is $\dim(\Herm(\X))-1=n^2-1$; the result now follows from \cite[Lemma 5]{frongillo2021elicitation}.
Thus, $\elicq(\id) = n^2-1$ by Proposition~\ref{prop:quantum-classical-elic-iden}.
As a corollary, since $\id$ refines all properties, we have $\elicq(\Gamma) \leq n^2-1$ for all quantum properties $\Gamma$.

Combining the above, we can now establish the complexity of strictly convex functions of a density matrix.
\begin{theorem}\label{thm:quantum-entropy-norms}
  For any strictly convex $F:\Dens(\X)\to\reals$, we have $\elicq(F) = n^2-1$.
\end{theorem}
\begin{proof}
  Let $F^\diamond:\Pfix\to\reals$, $F^\diamond : p \mapsto F(\phi^+ p)$.
  From Lemma~\ref{lem:convex-affine-composition} in \S~\ref{sec:omitted-proofs}, as $F$ is strictly convex, so is $F^\diamond$.
  We now have $\elicq(F) = \elicfix(F^\diamond) = \elicfix(\idfix) = n^2-1$ from Proposition~\ref{prop:classical-entropy-norm}, Corollary~\ref{cor:elic-complex-equiv}, and the above discussion.
\end{proof}

We saw in \S~\ref{sec:properties} that essentially all entropy functions fail to be elicitable.
The above shows that they also have maximal elicitation complexity, just as in the classical case.
(One would expect the same for most entanglement measures, by extension.)
Schatten norms for $\beta>1$, defined in \S~\ref{sec:properties}, are all strictly convex when raised to the power $\beta$.
\footnote{When $\beta=1$, the Schatten norm is the trace norm, which for density matrices is only interesting in its relative form, such as $\Gamma(\rho) = \|\rho - \rho_0\|_1$ for some fixed reference $\rho_0\in\Dens(\X)$.  We conjecture that the complexity of $\Gamma$ is still maximal; showing this statement would likely require a different proof technique.}
In essence, these results say that it is as hard to elicit the amount or value of information as it is to elicit the information itself.

Let us now turn to eigenvalues.
Formally, let us consider $\lambda:\Dens(\X)\to\reals^n$ as a quantum property.
As the function $F(\rho) = \inprod{\rho}{\rho} = \|\lambda(\rho)\|_2^2$ is strictly convex, Theorem~\ref{thm:quantum-entropy-norms} gives $\elicq(F) = n^2-1$.
But $\lambda$ refines $F$, so $\elicq(F) \leq \elicq(\lambda)$.
With the general upper bound $\elicq(\lambda) \leq n^2-1$ from the discussion above, we have the following.

\begin{corollary}
  $\elicq(\lambda) = n^2-1$.
\end{corollary}

A priori, this result may seem surprising.
A density matrix $\rho$ can be thought of as a transformed version of the classical distribution given by its eigenvalues $\lambda(\rho)$.
Since we can certainly elicit the full density matrix, as well as any classical distribution, it may seem reasonable that we should be able to elicit the classical distribution hiding inside a density matrix.
Yet doing so is as hard as eliciting the full density matrix itself.

\subsection{Top eigenvalues and eigenvectors}
\label{sec:complexity-top-eigenvalues}

As entropy, norms, and eigenvalues all have maximal elicitation complexity, we did not need a nontrivial upper bound; it cannot be higher than that of the density matrix itself.
To close, let us consider a more nuanced case, that of the top $k$ eigenvalues.
As we saw in \S~\ref{sec:properties}, the property $\lambda_{1..k}:\Dens(\X)\to\reals^k$, $(\lambda_{1..k}(\rho))_i = \lambda_i(\rho)$, is not elicitable.
Intuitively, since the case $k=n$ yields maximal complexity, one may expect the complexity of $\lambda_{1..k}$ to degrade gracefully in $k$.
We now use the main upper bound of \citet{frongillo2021elicitation} to show that indeed it does, with the caveat that we step outside the class of identifiable properties.

\citet[Theorem 3]{frongillo2021elicitation} show that if a scoring rule elicits a property, then the expected score of an optimal report can be elicited alongside that property.
We can directly translate the proof to our setting.
Let $S = (s,\mu)$ be a quantum score eliciting some property $\Gamma:\Dens(\X)\toto \R$.
We wish to (indirectly) elicit the property $\gamma:\Dens(\X)\to\reals$ given by $\gamma(\rho) = \max_{r\in\R} S(r;\rho)$.

Define a new quantum score $(s^*,\mu)$ with $s^*:(\reals\times\R)\times\Y\to\reals$ given by
$s^*((\alpha,r),y) = G(\alpha) + dG_\alpha (s(r,y) - \alpha)$,
where $G:\reals\to\reals$ is a strictly increasing and strictly convex function, and $dG_\alpha > 0$ is a subgradient $dG_\alpha \in \partial G(\alpha)$.
Then we have
\begin{align*}
  S^*((\alpha,r);\rho)
  &= \sum_{y\in\Y(r)} \inprod{\mu(r)_y}{\rho} ( G(\alpha) + dG_\alpha (s(r,y) - \alpha) )
  \\
  &= \sum_{y\in\Y(r)} \inprod{\mu(r)_y}{\rho} ( G(\alpha) + dG_\alpha (s(r,y) - \alpha) )
  \\
  &= G(\alpha) + dG_\alpha \left(\textstyle\sum_{y\in\Y(r)} \inprod{\mu(r)_y}{\rho} s(r,y) - \alpha\right)
  \\
  &= G(\alpha) + dG_\alpha (S(r;\rho) - \alpha)~.
\end{align*}
As $dG_\alpha > 0$, for all $\alpha$ we have that any $r\in\Gamma(\rho)$ maximizes the expected score.
Note that $\gamma(\rho) = S(r;\rho)$ for such an $r$.
By the subgradient inequality, $G(\gamma(\rho)) \geq G(\alpha) + dG_\alpha(\gamma(\rho)-\alpha)$, with equality (and thus optimality) if and only if $\alpha = \gamma(\rho)$.
Thus $S^*$ elicits the property $\hat\Gamma = (\gamma, \Gamma)$ given by $\hat \Gamma(\rho) = \{(\gamma(\rho),r) \mid r\in\Gamma(\rho)\}$.

Let us now apply this result to the top eigenvalue.
Recall from \S~\ref{sec:properties} the quantum score $S=(s,\mu)$ with $s(x,y) = \ones\{y = 1\}$ and $\mu(x) = \{I_\X - xx^*,xx^*\}$.
As we saw, $S$ elicits $\Gamma_1 : \Dens(\X) \toto \nsphere$, the set of unit-norm eigenvectors of maximum eigenvalue.
When $x$ is in fact an eigenvector with eigenvalue $\lambda_1(\rho)$, the expected score is therefore $S(x;\rho) = \inprod{xx^*}{\rho} = \inprod{x}{\rho x} = \inprod{x}{\lambda_1(x)x} = \lambda_1(\rho)$.
From the above, then, we can elicit the pair $(\lambda_1,\Gamma_1)$, the maximum eigenvalue along with a corresponding eigenvector.
Taking $G(\alpha) = \alpha^2$, we obtain the following natural score $s^*((\alpha,x),y) = 2\alpha\ones\{y=1\} - \alpha^2$.
This score is exactly binary-outcome Brier score, with the twist that $x$ actually determines the probability that $y=1$.

Generalizing to the top $k$ eigenvalues, let $\Gamma_k$ be the set-valued property specifying $k$ orthonormal vectors $x_1,\ldots, x_k$ with eigenvalues $\lambda_{1..k}(\rho)$.
Consider a report of the form $\alpha_1,\ldots,\alpha_k \geq 0$ (non-increasing) and $x_1,\ldots,x_k$ orthonormal.
Let $A = \sum_{i=1}^k \alpha_i \vvstar{x}{i}$ and note $\lambda(A) = \alpha \in \reals^n$ (padding with zeros after $\alpha_k$).
In fact, our score will only depend on $A$, so we may equivalently specify the report as any $A\in\Pos(\X)$ of rank at most $k$ and let $\alpha=\lambda(A)$ and $x_1,\ldots,x_n$ be any set of orthonormal vectors such that $A = \sum_{i=1}^n \alpha_i \vvstar{x}{i}$.
Take $S^* = (s^*,\mu)$ where $\mu(A) = \{\vvstar{x}{i}, \ldots, \vvstar{x}{n}\}$ and $s^*(A,y) = 2\alpha_y - \inprod{\alpha}{\alpha}$.
(Interestingly, $S^*$ is thus a spectral score.)
Then we have
\begin{equation*}
  S^*(A;\rho) = 2 \inprod{A}{\rho} - \inprod{\alpha}{\alpha} \leq 2 \inprod{\alpha}{\lambda(\rho)} - \inprod{\alpha}{\alpha} \leq \inprod{\lambda(\rho)}{\lambda(\rho)}~,
\end{equation*}
where these inequalities are tight if and only (i) there exists a unitary $U$ such that $A = U \Diag(\alpha) U^*$ and $\rho = U \Diag(\lambda(\rho)) U^*$, and (ii) $\alpha = \lambda(\rho)$.
Put together, these conditions are equivalent to $x_1,\ldots,x_k$ being eigenvectors of $\rho$ with eigenvalues $\alpha_1=\lambda_1(\rho),\ldots,\alpha_k=\lambda_k(\rho)$.
In other words, $S^*$ elicits the pair $(\lambda_{1..k},\Gamma_{1..k})$.

One can specify a rank-$k$ Hermitian matrix with $2nk - k^2$ real-valued parameters.
\footnote{Because the eigenvectors $x_1,\ldots,x_k$ can each be independently multiplied by any phase $e^{\theta i}$ and remain eigenvectors, it takes only $2nk - k^2 - k$ real parameters to specify them, and another $k$ to specify $\alpha$.}  
It may therefore seem at this point that we have shown $\elicq(\lambda_{1..k}) \leq 2nk - k^2$.
In fact, this statement is likely false, because $\Gamma_{1..k}$ is not identifiable: restricted to diagonal matrices $\rho = \Diag(\lambda)$, $\Gamma_{1..k}$ gives the top $k$ most likely outcomes, which is not an identifiable property (\S~\ref{sec:discussion-finite-properties}).
Nonetheless, we would have the complexity upper bound of $2nk - k^2$ for a more flexible class of properties that contains $(\lambda_{1..k},\Gamma_{1..k})$.

\section{Discussion and Future Directions}
\label{sec:discussion}

This article attempts to lay a foundation of a much broader theory of quantum information elicitation.
Some of these directions have already been explored, such as in machine learning.
Others, particularly in more economic settings, remain open.

\subsection{Machine learning of quantum mixed states}
\label{sec:discussion-ML-mixed}

The phrase ``quantum machine learning'' has come to mean essentially four settings, depending on whether the data is classical or quantum, and whether the learning algorithm is classical or quantum~\cite{aimeur2006machine}.
Closest to the present work is the (quantum, classical) pair, i.e., classical learning algorithms applied to data generated from a quantum process.
One setting is quantum tomography and variations, where one wishes to recover a quantum state, or a function from some feature vector to a quantum state.
This literature is related to recent work in online learning, where one predicts a sequence of mixed states in an online fashion~\cite{koolen2011learning}.

Several loss functions (presented here as positive-oriented scores) have been proposed in the machine learning literature for learning from quantum mixed states.
A partial list follows.

\smallskip
\begin{tabular}{ll}
  $S_1(\rho';\rho) = \inprod{\log \rho'}{\rho}$
  &
    Log spectral score, also ``matrix entropic loss'' \cite{koolen2011learning}
  \\
  $S_2(\rho';\rho) =  \log\det( (\rho')^{-1}\rho) - \inprod{(\rho')^{-1}}{\rho}$
  &
    Log determinent score~\cite{stein1961estimation,gneiting2011making}
  \\
  $S_3(\rho';\rho) = \inprod{\rho'}{\rho}$
  &
    Trace score \cite{koolen2011learning}
  \\
  $S_4(\rho';\rho) = \log \inprod{\rho'}{\rho}$
  &
    Log trace score \cite{koolen2011learning}
  \\
  $S_5(\rho';\rho) = \log \Tr \exp( \log \rho' + \log \rho)$
  &
    A variation on the log trace score~\cite{koolen2011learning}
\end{tabular}
\vskip2pt

\noindent
Interestingly, only $S_1$ and $S_2$ are truthful; $S_2$ is of the form~\eqref{eq:main-char} for $F(\rho) = -\log\det\rho$.
The fact that $S_5$ is not truthful answers a question posed by \citet{koolen2011learning}.
\footnote{Take $\rho=\Diag(\lambda)$ and $\rho'=\Diag(\lambda')$.  Then $S_5(\rho';\rho) = \log\sum_i \lambda_i\lambda'_i = S_3(\rho';\rho)$; the optimal $\lambda'$ puts all weight on $\lambda_1$.}
Even more importantly, the scores $S_4$ and $S_5$ are not \emph{physically implementable} (or \emph{operational}~\cite{blume2010optimal}), meaning there is no way to directly compute the score using a measurement, or in some cases, even any finite sequence of measurements.
Verifying whether a score is physically implementable is straightforward from \S~\ref{sec:quant-scor-rules}: the score must be (extended) linear in the true mixed state $\rho$.

\subsection{Restricting to pure states}
\label{sec:discussion-pure}

We have focused on the case where the underlying belief is a mixed quantum state.
It is also interesting to design quantum scores when one assumes the true state $\rho$ is pure, that is, $\rho = xx^*$ for some $x\in\X$.
The set of pure states is the relative boundary $\partial \Dens(\X)$ of the set of density matrices.
The affine score characterization (\S~\ref{sec:class-scor-rules}) therefore states that any quantum score which is truthful for pure states can be extended to a truthful score on all of $\Dens(\X)$; in that sense, there are no ``tricks'' that allow one to take advantage of the restricted setting.

Nonetheless, some particularly simple forms arise when considering pure states directly.
For example, the simple trace score $S_3(\rho';\rho) = \inprod{\rho'}{\rho}$, which is not truthful for mixed states, is strictly truthful for pure states: $S_3(xx^*;zz^*) = |\inprod{x}{z}|^2$, which is maximized by $x=z$.
Indeed, this score has been used in learning algorithms from pure states~\cite{lee2018learning}.

\subsection{Extensions of finite properties}
\label{sec:discussion-finite-properties}

The trace score $S_3$ also happens to coincide with the score for learning a maximum eigenvector $\Gamma_1$ given in \S~\ref{sec:properties}.
The optimal report for $S_3(\rho';\rho) = \inprod{\rho'}{\rho}$ over all $\rho'\in\Dens(\X)$ is $\rho' = xx^*$ where $x$ is an eigenvector of $\rho$ of maximum eigenvalue.
This statement also shows why $S_3$ is not truthful.
As such, $S_3$ can be thought of as the quantum analog of 0-1 loss (classification error) from machine learning, which elicits the mode (most likely outcome).
The two coincide when $\rho$ and $\rho'$ are diagonal matrices: the optimal report is $\rho' = \Diag(e_i)$ where $e_i$ is the $i$th standard basis vector and $i \in \argmax_j \lambda_j(\rho) = \textrm{mode}(\lambda)$.
(As the mode is not identifiable, this shows why $\Gamma_1$ is not either.)

Other extensions of finite properties are interesting to consider as well.
The classical score for top-$k$ classification is $\ones\{y\in T\}$ where one predicts a set $T$ of size $k$.
Similar to $\Gamma_1$ from the mode (the special case $k=1$), the score given for $\Gamma_{1..k}$ in \S~\ref{sec:properties} is the natural analog of this top-$k$ score.
As one final example, consider the problem of classification with an abstain option~\cite{ramaswamy2018consistent}.
Classically, the report is some $r \in \Y \cup \{\bot\}$ and the score is $\alpha \in$ if $r=\bot$ and $\ones\{r=y\}$ otherwise, for some $0 < \alpha < 1$.
(This score incentivizes the agent/algorithm to abstain ($\bot$) when no outcome has at least $\alpha$ probability.)
Modifying the score for $\Gamma_1$, we can take $s(\bot,y) = 1/2$ and $s(x,y) = \ones\{y=1\}$ otherwise, with $\mu(x) = \{I_\X - xx^*, xx^*\}$ as before.
This score elicits an eigenvector of $\lambda_1(\rho)$ if $\lambda_1(\rho) > \alpha$, and $\bot$ otherwise.

\subsection{Multi-agent elicitation: wagering mechanisms and prediction markets}

A wagering mechanism seeks to simultaneously elicit information from a group of agents~\cite{lambert2015axiomatic}.
In the simplest setting, we have $m$ agents who will report some $\rho_1,\ldots,\rho_m\in\Dens(\X)$.
Given a quantum score $S=(s,\mu)$, in the quantum analog of the weighted-score wagering mechanism (with unit wagers for simplicity), the expected payoff to agent $i$ is
\begin{equation}
  \label{eq:quantum-wagering}
  \text{Payoff}_i(\rho_1,\ldots,\rho_m;\rho) = S(\rho_i;\rho) - \frac 1 {m-1} \sum_{j\neq i} S(\rho_j;\rho)~.
\end{equation}
There is an issue with this form, however: the mechanism must choose a single measurement, yet there is no guarantee that $\mu(\rho_i)$ is the same for all $i$.
A simple remedy, of course, is to require $S$ to be a fixed-measurement score (\S~\ref{sec:fixed-measurement-scores}).
An interesting open question is how to design a quantum wagering mechanism when the measurement must be projective.
Here one would imagine choosing the measurement as a function of all the reports, $\mu(\rho_1,\ldots,\rho_m) \in \Meas_\Y(\X)$,
though the analysis of \citet{blume2006accurate} suggests care is needed.

The situation for prediction markets is similar.
In the analog of the classical scoring rule market~\cite{hanson2003combinatorial,frongillo2018axiomatic}, there is a global density matrix, initialized to some $\rho_0\in\Dens(\X)$, and at time $t$ a trader who modifies this global matrix $\rho_{t-1} \to \rho_t$ gains the following in expectation
\begin{equation}
  \label{eq:quantum-wagering}
  \text{Payoff}_i(\rho_{t-1}\to\rho_t;\rho) = S(\rho_t;\rho) - S(\rho_{t-1};\rho)~.
\end{equation}
Interestingly, cost-function-based market makers~\cite{abernethy2013efficient} also have natural analog: maintain a share vector $Q\in\Herm(\X)$; a bundle of shares $R\in\Herm(\X)$ costs $F^*(Q+R) - F^*(Q)$ and pays off $\inprod{R}{\rho}$, where $F^*$ is the convex conjugate of some convex $F:\Dens(\X)\to\reals$.
A natural choice for both is to take $F=-H$, negative von Neumann entropy, yielding the quantum analog of LMSR~\cite{hanson2003combinatorial,hanson2007logarithmic}, where $F^*(Q) = \log\sum_i\exp\lambda_i(Q)$.

In both of these settings, just as with wagering mechanisms, the relevant scores and payouts are only physically implementable when using the same measurement for all agents.
We would like to understand if the measurement can be projective and chosen based on the market history in a way which preserves truthfulness.
Unlike in the wagering mechanism setting, however, where a natural desideratum was for the score to be bounded, in a prediction market there are reasons to want an unbounded score, so that the market cannot get ``stuck''.
In particular, we ideally would like $F$ to be of Legendre type, so like von Neumann entropy its gradients blow up at the boundary of $\Dens(\X)$, but there is no way to achieve this behavior with a fixed-measurement quantum score (Proposition~\ref{prop:projective-expressive} and preceeding discussion).

\subsection{Quantum information and Bayes}
\label{sec:quant-inform-bayes}

Finally, a natural question from an economic point of view is to what extent one can talk about Bayesian updates of a belief which is a quantum mixed state.
If one currently believes $\rho_0$, but observes some signals, perhaps measurements of a copy of the true $\rho$, how should one incorporate this information into a new belief?
This question has been studied from a variety of vantage points; see e.g.\ \cite[Theorem 4]{blume2006accurate} and \cite{warmuth2005bayes,warmuth2010bayesian}.
With such a theory in hand, one could develop analogs of peer prediction and costly information aquisition.

\subsection*{Acknowledgements}
The author gratefully acknowledges many conversations and insights from
Ian Kash,
Steven Kordonowy,
Anthony Polloreno,
Graeme Smith,
and
Bo Waggoner.
This material is based upon work supported by the National Science Foundation under Grant No.\ IIS-2045347.

\newpage
\bibliographystyle{plainnat}
\bibliography{refs-quantum,refs-elic}

\newpage
\appendix

\section{Subteties Arising from the Extended Reals}
\label{sec:app-ext-reals}

\subsection{Extended linear functions}
\label{sec:extend-linear-functions}

Let $\V$ be a real vector space.
A real-valued function $\ell:\V\to\reals$ is \emph{linear} if $\ell(\alpha x + x') = \alpha\ell(x)+\ell(x')$ for all $x,x'\in\V$ and $\alpha\in\reals$.
It turns out this definition does not immediately generalize to extended-real-valued functions.
Following \citet{waggoner2021linear}, we say a function $\ell:\V\to\extreals$ is \emph{extended linear}, denoted $\ell \in \extlin{\V}$, if two conditions hold:
\begin{itemize}
\item[(i)] $\ell(\alpha v) = \alpha \ell(v)$ for all $\alpha \in \reals$ and $v\in\V$;
\item[(ii)] $\ell(v + v') = \ell(v) + \ell(v')$ for all $v,v'\in\V$ such that $\{\ell(v),\ell(v')\} \neq \{\infty,-\infty\}$.
\end{itemize}
As usual, one defines $0\cdot\infty = 0\cdot(-\infty) = 0$.

\subsection{Extended vectors and extended Hermitian matrices}
\label{sec:extend-herm}

To accommodate classical scoring rules which assign values of $-\infty$, one needs to consider the set of extended linear functions $\lin(\reals^\Y\to\extreals)$.
While in general extended linear functions can be quite complex to describe~\cite{waggoner2021linear}, fortunately, as observed by \citet{gneiting2007strictly}, for scoring rules we will only require those of a simple form: inner products of what we call \emph{extended vectors} $v \in (\reals\cup\{-\infty\})^\Y$.
Given such a vector $v$, we define the inner product $\inprod{v}{x}$ for $x\in\reals^\Y$ by $\inprod{v}{x} = \sum_{y: v_y \in \reals} v_y x_y -\infty \sum_{y : v_y = -\infty} x_y$.
Note that by grouping the terms according to the value of $v$ first, we avoid expressions of the form $\infty - \infty$.
Moreover, one can check that the function $\ell_v : x\mapsto \inprod{v}{x}$ is an element of $\lin(\reals^\Y\to\extreals)$.
As the expected score for a given regular scoring rule $\hat s$ is given by $\E_p \hat s(p',Y) = \inprod{s(p',\cdot)}{p}$, it therefore suffices to work with extended linear functions given by extended vectors.

Analogously, to accommodate \emph{quantum} scoring rules which assign values of $-\infty$, we need to consider extended linear functions from $\Herm(\X)$ to $\extreals$.
Just as in the classical case, fortunately, we will only require those of a simple form: inner products of what we call \emph{extended Hermitian} matrices $E\in\extHerm(\X)$.

Formally, $\extHerm(\X)$ denotes the set of formal expressions of the form $A - \infty B$ where $A\in\Herm(\X)$, $B\in\Pos(\X)$, and $AB=0$.
For any such expression, we define $\inprod{A-\infty B}{X} = \inprod{A}{X} - \infty \inprod{B}{X}$ for any $X\in\Herm(\X)$.
Thus, given $E \in \extHerm$, the function $\ell : X \mapsto \inprod{E}{X}$ is extended linear, i.e., $\ell \in \lin(\Herm(\X)\to\extreals)$.
Now let an orthonormal basis $\{x_1,\ldots,x_n\}$ of $\X = \C^n$ and vector $v \in (\reals\cup\{-\infty\})^n$ be given, and let $U\in\lin(\X\to\X)$ be the unitary matrix with columns $x_1,\ldots,x_n$.
Then we may write
\begin{equation}
  \label{eq:extherm-diag}
  U \Diag(v) U^* := \sum_{y=1}^n v_y \vvstar{x}{y} = \sum_{y:v_y\in\reals} v_y \vvstar{x}{y} - \infty \sum_{y:v_y=-\infty} \vvstar{x}{y} \in\extHerm(\X)~.
\end{equation}
Moreover, given $E \in \extHerm(\X)$, as $E = A - \infty B$ for $AB = 0$, there exists a unitary matrix $U$ simultaneously diagonalizing $A$ and $B$, and we may define $\lambda(E) \in (\reals \cup\{-\infty\})^n$ by the above, reordering the elements of $v$ appropriately.

The following result justifies our focus on $\extHerm(\X)$; as we will see later, it implies that $\extHerm(\X)$ captures all extended linear functions arising from quantum scores.

\begin{lemma}\label{lem:extherm}
  Let $A_1,\ldots,A_m \in \Pos(\X)$, and $\alpha_1,\ldots,\alpha_m\in\reals\cup\{-\infty\}$ be given, and consider $\ell \in \lin(\Herm(\X)\to\extreals)$ such that $\ell(\rho) = \sum_{i=1}^m \inprod{A_i}{\rho} \alpha_i$ for all $\rho\in\Dens(\X)$.
  Then there exists $E\in\extHerm(\X)$ such that $\ell(\rho) = \inprod{E}{\rho}$ for all $\rho\in\Dens(\X)$.
\end{lemma}
\begin{proof}
  Let
  \begin{align}
    \label{eq:X_S}
    A' &= \sum_{i:\alpha_i\in\reals} \alpha_i A_i \in \Herm(\X)
    \\
    \label{eq:X_Splus}
    B &= \sum_{i:\alpha_i=-\infty} A_i \in \Pos(\X)~.
  \end{align}
  For all $\rho\in\Dens(\X)$, we have $\ell(\rho) = \inprod{A'}{\rho} - \infty \inprod{B}{\rho}$.
  Let $x_1,\ldots,x_n$ be an orthonormal eigenbasis for $B$, such that $x_1,\ldots,x_k \in \ker B$, and $B = \sum_{y = k+1}^n \lambda_y \vvstar{x}{y}$ for some eigenvalues $\lambda_y > 0$.
  Define $A = A' - \sum_{y = k+1}^n A' \vvstar{x}{y}$.
  For all $y \in \{k+1,\ldots,n\}$, we have $A x_y = A' x_y - \left(\sum_{j = k+1}^n A' x_j x_j^*\right) x_y = A' x_y - A' x_y = 0$ by orthonormality, giving $x_y \in \ker A$.
  Thus, $A B = 0$, and $E = A - \infty B \in \extHerm(\X)$.

  It remains to show that $\inprod{E}{\cdot}$ and $\ell$ agree on $\Dens(\X)$.
  Let $\rho\in\Dens(\X)$.
  If $\inprod{B}{\rho} \neq 0$, then as $B,\rho\in\Pos(\X)$ we must have $\inprod{B}{\rho} > 0$ and thus $\inprod{E}{\rho} = \inprod{A}{\rho} -\infty\inprod{B}{\rho} = -\infty = \inprod{A'}{\rho} -\infty\inprod{B}{\rho} = \ell(\rho)$.
  Otherwise, we have $\inprod{B}{\rho} = 0$, giving $B \rho = 0$ from Lemma~\ref{lem:pos-zero-inner-prod}.
  As $x_1,\ldots,x_k \in \ker B$, we must have $x_{k+1},\ldots,x_n \in \ker \rho$.
  We have $\inprod{A'}{\rho} = \sum_{y=1}^n x_y^* A' \rho x_y = \sum_{y=1}^k x_y^* A' \rho x_y = \sum_{y=1}^k x_y^* A \rho x_y = \inprod{A}{\rho}$, where we used orthonormality and the definition of $A'$ in the third equality.
  Hence, $\inprod{E}{\rho} = \inprod{A}{\rho} = \inprod{A'}{\rho} = \ell(\rho)$.
\end{proof}

\subsection{Fixed-measurement scores and the relative boundary}

\begin{lemma}\label{lem:fixed-meas-boundary}
  Let $\mufix\in\Meas_\Y(\X)$ be any tomographically complete measurement over a finite $\Y$.
  Assume $\mufixy \neq 0$ for all $y\in\Y$ without loss of generality.
  Then there exists $\rho$ on the relative boundary of $\Dens(\X)$ such that $\inprod{\mufix}{\rho}$ is in the relative interior of $\simplex$.
\end{lemma}
\begin{proof}
  For each $y\in\Y$, let $x_y$ be an eigenvector of $\mufixy$ of maximum eigenvalue $\lambda(\mufixy)_1 > 0$.
  Let $x_y^\bot$ be the set of vectors in $\X$ orthogonal to $x_y$.
  As $x_y^\bot$ has measure zero in $\X$, so does $\cup_{y\in\Y} x_y^\bot$.
  Thus, there exists some $x\in\X$ such that $\inprod{x_y}{x} \neq 0$ for all $y\in\Y$.
  Let $\rho = xx^*$, which is on the relative boundary of $\Dens(\X)$.
  Now $\inprod{\mufix}{\rho}_y = \inprod{\mufixy}{xx^*} \geq \lambda(\mufixy)_1 \inprod{\vvstar{x}{y}}{xx^*} = \lambda(\mufixy)_1 |\inprod{x_y}{x}|^2 > 0$.
  Thus, $\inprod{\mufix}{\rho}$ is in the relative interior of $\simplex$.
\end{proof}

\section{Omitted Proofs}
\label{sec:omitted-proofs}

\begin{lemma}\label{lem:convex-affine-composition}
  Let $g:C\to\reals$ and $h:C'\to\reals$ where $C\subseteq\V$ and $C'\subseteq\V'$ for real vector spaces $\V,\V'$, and $\phi \in \lin(\V\to\V')$ such that $C' = \phi(C)$ and $g = h \circ \phi$.
  Then $g$ is convex if and only if $h$ is convex, and $g$ is strictly convex if and only if $\phi$ is injective on $C$ and $h$ is strictly convex.
\end{lemma}
\begin{proof}
  Convexity of $h$ implies convexity of $g$~\cite[Proposition 2.1.4]{hiriart2012fundamentals}.
  The proof shows that strict convexity of $h$ implies strict convexity of $g$ when $\phi$ is injective.

  For the converse, let $\phi^+$ be the Moore--Penrose pseudoinverse of $\phi$.
  For any $x'\in C'$ we have $x' = \phi x$ for some $x\in C$, and thus $g(\phi^+ x') = h(\phi \phi^+ \phi x) = h(\phi x) = h(x')$.
  Hence $h = g \circ \phi^+$, so $g$ being convex implies the same for $h$.
  Moreover, as $\phi^+$ is always injective on $C'$, strict convexity of $g$ implies the same for $h$.  
\end{proof}

\fixedmeaschar*
\begin{proof}
  Let $S$ be truthful, and let $F(\rho) = S(\rho;\rho)$.
  By Lemma~\ref{lem:fixed-meas-expected-score}, we may write $F(\rho) = f(\inprod{\mufix}{\rho})$ for some $f:\simplex\to\reals$.
  By Lemma~\ref{lem:convex-affine-composition}, $F$ is convex if and only if $f$ is convex.
  By the same lemma, $F$ is strictly convex if and only if $f$ is strictly convex and $\phi:\rho\mapsto\inprod{\mufix}{\rho}$ is injective; the latter condition is equivalent to $\mufix$ being tomographically complete.
  That the given form is truthful follows from Theorem~\ref{thm:main-char} and the fact that $F(\rho) := f(\inprod{\mufix,\rho})$ is convex with the given form of the subgradient~\cite[Theorem 4.2.1]{hiriart2012fundamentals}.
\end{proof}

\lemspecsubgrad*
\begin{proof}
  First suppose $U \Diag(d) U^* \in \extsub F(U \Diag(\lambda) U^*)$.
  Let $\lambda'\in\simplex$ be arbitrary.
  Applying the subgradient inequality, we have
  \begin{align*}
    f(\lambda') = F(U \Diag(\lambda') U^*)
    &\geq F(U \Diag(\lambda) U^*) + \inprod{U \Diag(d) U^*}{U \Diag(\lambda') U^* - U \Diag(\lambda) U^*}
    \\
    &= f(\lambda) + \inprod{\Diag(d)}{\Diag(\lambda') - \Diag(\lambda)}
    \\
    &= f(\lambda) + \inprod{d}{\lambda' - \lambda}~.
  \end{align*}
  As $\lambda'$ was arbitrary, we have $d \in \extsub f(\lambda)$.
  
  For the converse, we first observe that if the entries of $\lambda$ are in non-increasing order, so are the entries of $d \in \extsub f(\lambda)$, where we think of $d$ as an element of $(\reals\cup\{-\infty\})^n$; see \S~\ref{sec:extend-herm}.
  To see this, observe first that if $d_y = -\infty$ we must have $\lambda_y = 0$; otherwise, the subgradient inequality $f(\lambda') \geq f(\lambda) + \inprod{d}{\lambda'-\lambda}$ could not hold for any $\lambda' \in \dom f$.
  Thus, there exists some integer $k$ such that $d_y \in \reals$ for all $y\in\{1,\ldots,k\}$, and $d_y = -\infty$ and $\lambda_y=0$ for all $y \in \{k+1,\ldots,n\}$.
  Restricting $f$, $\lambda$, and $d$ to the first $k$ coordinates, we can apply \cite[Theorem 3.1]{lewis1996convex} to see that the first $k$ coordinates of $d$ are non-decreasing.
  Since the smallest elements are $-\infty$, we see that $d$ has non-decreasing entries.

  Now let $\rho,\rho'\in\Dens(\X)$,
  $d\in\extsub f(\lambda(\rho))$, and $U\in\lin(\X\to\X)$ unitary such that $\rho = U \Diag(\lambda(\rho)) U^*$.
  Let $k$ be the integer above for $d$ and $\lambda(\rho)$, and let $\hat d \in \reals^k$ be the first $k$ entries of $d$.
  Let $dF_\rho = U \Diag(d) U^*$.
  By the above, we have $\inprod{dF_\rho}{\rho} = \inprod{\hat d}{\lambda(\rho)_{1..k}} = \inprod{d}{\lambda(\rho)} \in \reals$.

  Now consider $\inprod{dF_\rho}{\rho'}$.
  By the proof of Lemma~\ref{lem:extherm}, we may write $dF_\rho = A - \infty B$ where $AB=0$ and $B\in\Pos(\X)$.
  If $\inprod{B}{\rho'} > 0$, then $\inprod{dF_\rho}{\rho'} = -\infty$.
  Otherwise, $\inprod{B}{\rho'} = 0$, and we must have $B \rho' = 0$ by Lemma~\ref{lem:pos-zero-inner-prod}.
  In particular, $(U^* \rho' U)_{yy'} = 0$ if $d_y = -\infty$ or $d_{y'} = -\infty$.
  Let $C$ be the first $k\times k$ block of $(U^* \rho' U)$.
  Then $\inprod{dF_\rho}{\rho'} = \inprod{\Diag(\hat d)}{C} \geq \inprod{\hat d}{\lambda(\rho')_{1..k}} = \inprod{d}{\lambda(\rho')}$, where the inequality follows from Lemma~\ref{lem:pos-zero-inner-prod}.

  Combining the above, we have
  \begin{align*}
    F(\rho')
    &= f(\lambda(\rho'))
    \\
    &\geq f(\lambda(\rho)) + \inprod{d}{\lambda(\rho')} - \inprod{d}{\lambda(\rho)}
    \\
    &\geq f(\lambda(\rho)) + \inprod{dF_\rho}{\rho'} - \inprod{d}{\lambda(\rho)}
    \\
    &= F(\rho) + \inprod{dF_\rho}{\rho'} - \inprod{dF_\rho}{\rho}
    \\
    &= F(\rho) + \inprod{dF_\rho}{\rho' - \rho}~.
  \end{align*}
  As $\rho$ was arbitrary, we conclude $dF_\rho \in \extsub F(\rho)$.
\end{proof}

The following is an analog of \citet[Corollary 3.3]{lewis1996convex} for the extended linear case.
\lemFfconvex*
\begin{proof}
  As in the proof of Lemma~\ref{lem:spectral-convex-subgradients}, $F$ being (strictly) convex immediately implies the same of $f$.
  For the converse, Lemma~\ref{lem:spectral-convex-subgradients} gives that $F$ is convex if $f$ is.
  To show strictness, suppose that $f$ is strictly convex.
  By \cite[Lemma 3.11]{waggoner2021linear},
  a convex function is strictly convex if and only if the (extended) subgradient inequality between $x,y$ is strict unless $x=y$.
  Accordingly, let $\rho,\rho'\in\Dens(\X)$, $\rho\neq \rho'$, and let $U \Diag(d) U^*\in\extsub{F}(\rho)$.
  By the subgradient inequality and Lemma~\ref{lem:inner-prod-eigenvalues}, we have
  \begin{align*}
    F(\rho') - F(\rho)
    &= f(\lambda(\rho')) - f(\lambda(\rho))
    \\
    &\geq \inprod{d}{\lambda(\rho')-\lambda(\rho)}
    \\
    &\geq \inprod{U \Diag(d) U^*}{\rho'-\rho}~.
  \end{align*}
  If $\lambda(\rho) \neq \lambda(\rho')$, the first inequality is strict by strict convexity of $f$.
  Otherwise $\lambda(\rho)=\lambda(\rho')$, and we must have $\inprod{U \Diag(d) U^*}{\rho'} < \inprod{d}{\lambda(\rho')}$ from Lemma~\ref{lem:inner-prod-eigenvalues}, so the second inequality is strict.
\end{proof}

\propeliciden*
\begin{proof}
  If some classical scoring rule $\hat s:\R\times\Y\to\reals$ elicits $\Gammafix$, we can simply take the fixed-measurement quantum score $\Sfix = (\hat s,\mu:\R\mapsto\mufix)$.
  Then $\Sfix(r;\rho) = \hat s(r;\phi \rho)$.
  Thus, for all $\rho\in\Dens(\X)$,
  \begin{equation*}
    \argmin_{r\in\R} S(r;\rho)
    = 
    \argmin_{r\in\R} \hat s(r;\phi \rho)
    = 
    \Gammafix(\phi \rho)
    =
    \Gamma(\phi^+ \phi \rho)
    =
    \Gamma(\rho)~,
  \end{equation*}
  so $\Sfix$ elicits $\Gamma$.

  Conversely, let $S$ be a quantum score eliciting $\Gamma$.
  Define $\hat s(r,y) = (\phi^{+*} Z_S(r))_y \in \reals$.
  Then we have
  \begin{align*}
    \hat s(r;\phi \rho)
    &= \sum_{y\in\Y} (\phi \rho)_y \hat s(r,y)
    = \sum_{y\in\Y} (\phi \rho)_y (\phi^{+*} Z_S(r))_y
    \\
    &= \inprod{\phi^{+*} Z_S(r)}{\phi \rho}
    = \inprod{Z_S(r)}{\phi^+ \phi \rho}
    = \inprod{Z_S(r)}{\rho}
    = S(r;\rho)~.
  \end{align*}
  Just as above, for all $p\in\Pfix$,
  \begin{equation*}
    \argmin_{r\in\R} \hat s(r;p)
    =
    \argmin_{r\in\R} \hat s(r;\phi \phi^+ p)
    = 
    \argmin_{r\in\R} S(r;\phi^+ p)
    = 
    \Gamma(\phi^+ p)
    = 
    \Gammafix(p)~,
  \end{equation*}
  so $\hat s$ elicits $\Gammafix$.

  The argument for identifiability follows similarly.
  Suppose $v:\R\to\reals^{k\times\Y}$ is an identification function for $\Gammafix$.
  Define $V:\R\to\Herm(\X)^k$ by $V(r)_i = \phi^* v(r)_i$, where $v(r)_i$ is the $i$th row of $v(r)$.
  Then for all $r\in\R$, $\rho\in\Dens(\X)$, and $i \in \{1,\ldots,k\}$,
  \begin{equation*}
    \inprod{V(r)_i}{\rho} = 0
    \iff
    \inprod{\phi^* v(r)_i}{\rho} = 0
    \iff
    \inprod{v(r)_i}{\phi \rho} = 0~.
  \end{equation*}
  Thus, $\inprod{V(r)}{\rho} = 0 \iff
    \Gammafix(\phi \rho) = r
    \iff
    \Gamma(\rho) = r$.

  Conversely, if $V:\R\to\Herm(\X)$ is an identification function for $\Gamma$, let $v(r)_i = \phi^{+*} V(r)_i$.
  Then for all $r\in\R$, $p\in\Pfix$, and $i \in \{1,\ldots,k\}$,
  \begin{equation*}
    \inprod{v(r)_i}{p} = 0
    \iff
    \inprod{\phi^{+*}V(r)_i}{p} = 0
    \iff
    \inprod{V(r)_i}{\phi^+ p} = 0~.
  \end{equation*}
  Again, we have $\inprod{v(r)}{p} = 0 \iff \Gamma(\phi^+ p) = r \iff \Gammafix(p) = r$.
\end{proof}

\end{document}